\documentclass[cmp]{svjour}
\pdfoutput=1
\newenvironment{definition*}[1][Definition]{\begin{trivlist}
\item[\hskip \labelsep {\bfseries #1.}]}{\end{trivlist}}

\usepackage{amssymb} 
\usepackage{amsfonts}
\usepackage{amsmath}
\usepackage{latexsym}

\usepackage{ifthen}
\usepackage{graphicx}
\usepackage{verbatim}
\usepackage{fix-cm}
\usepackage{hyperref} 
\usepackage{stmaryrd}
\usepackage{xspace}
\usepackage{graphicx}
\graphicspath{{figures/}}
\usepackage{xy}
\xyoption{all}
\usepackage{color}
\usepackage{subfigure} 

\usepackage{tikzfig}

\pgfdeclarelayer{edgelayer}
\pgfdeclarelayer{nodelayer}
\pgfsetlayers{edgelayer,nodelayer,main}

\newcommand{\applyop}[2]{\ensuremath{#1\left(\tikz[baseline=(current bounding box).east]{\path [use as bounding box] (0,0) rectangle #2;}\right)\phantom{#1}}}

\tikzstyle{white dot}=[dot,fill=white]
\tikzstyle{none}=[inner sep=0mm]

\tikzstyle{dot}=[inner sep=0.5mm,fill=black,draw=black,shape=circle]
\tikzstyle{dotpic}=[baseline=-0.25em,shorten <=-0.1mm,shorten >=-0.1mm,scale=0.6]
\tikzstyle{small}=[inner sep=0.4mm]
\tikzstyle{dotpic inline}=[baseline=(current bounding box).south]
\tikzstyle{cdiag}=[baseline=(current bounding box).east,xscale=1.5,-latex]
\tikzstyle{every loop}=[]
\tikzstyle{rn}=[dot]
\tikzstyle{gn}=[dot]
\tikzstyle{bn}=[inner sep=0pt]
\tikzstyle{uploop}=[in=45,out=135,loop]
\tikzstyle{downloop}=[in=-45,out=-135,loop]
\tikzstyle{small dot}=[dot,inner sep=0.4mm]
\tikzstyle{small white dot}=[small dot,fill=white]
\tikzstyle{small gray dot}=[small dot,fill=gray!50]
\tikzstyle{greenbox}=[rectangle,fill=gray!30,draw=gray!50!black,minimum height=7mm,minimum width=7mm]
\tikzstyle{bluebox}=[rectangle,fill=white,draw=gray,minimum height=7mm,minimum width=7mm]
\tikzstyle{cnot}=[fill=white,shape=circle,inner sep=-1.4pt]
\tikzstyle{pt}=[regular polygon,regular polygon sides=3,draw=black,scale=0.75,inner sep=-0.5pt]
\tikzstyle{copt}=[pt,regular polygon rotate=180]
\tikzstyle{tick}=[sloped,rotate=90,font=\small\bf,xshift=0.07mm]
\tikzstyle{white dot}=[dot,fill=white]
\tikzstyle{gray dot}=[dot,fill=gray!50]
\tikzstyle{gs dot}=[dot,fill=gray]
\tikzstyle{small dotpic}=[dotpic,scale=0.6]
\tikzstyle{mux}=[rectangle,draw=black,scale=0.5,minimum width=1.8cm,minimum height=1cm]

\tikzstyle{square box}=[rectangle,fill=white,draw=black,minimum height=6mm,minimum width=6mm]
\tikzstyle{square gray box}=[rectangle,fill=gray!30,draw=black,minimum height=6mm,minimum width=6mm]
\tikzstyle{dashed box}=[draw=black,dashed,minimum height=12mm,minimum width=12mm,fill=gray!20]
\tikzstyle{box vertex}=[rectangle,draw=black]

\usepackage{url}

\newcommand{\Q}{\ensuremath{{\mathbb C^2}}}
 
\newcommand{\bra}[1]{
    \ensuremath{\left\langle #1 \right|}\xspace}
\newcommand{\ket}[1]{
   {\left|  #1 \right\rangle}\xspace}
\newcommand{\braket}[2]{\langle#1|#2\rangle}

\def\PICT{\begin{center}{\Large Picture:} }
\def\EPICT{\end{center}}

\newcommand{\inlinegraphic}[2]{
  \dimendef\grafheight=255\dimendef\grafvshift=254
  \grafheight=#1
  \grafvshift=-0.5\grafheight
  \advance\grafvshift by 0.5ex
  \raisebox{\grafvshift}{\includegraphics[height=\grafheight]{#2}\xspace}
}

\newcommand{\dotunit}[1]{%
\begin{tikzpicture}[dotpic,yshift=-1mm]
\node [#1] (0) at (0,0.25) {}; 
\draw (0)--(0,-0.15);
\end{tikzpicture}\,}
\newcommand{\dotcounit}[1]{%
\begin{tikzpicture}[dotpic,yshift=1mm]
\node [#1] (0) at (0,-0.25) {}; 
\draw (0)--(0,0.15);
\end{tikzpicture}\,}
\newcommand{\dotmult}[1]{%
\begin{tikzpicture}[dotpic]
	\node [#1] (a) {};
	\draw (a) -- (-90:0.35);
	\draw (a) -- (45:0.4);
	\draw (a) -- (135:0.4);
\end{tikzpicture}}
\newcommand{\dotcomult}[1]{%
\begin{tikzpicture}[dotpic]
	\node [#1] (a) {};
	\draw (a) -- (90:0.35);
	\draw (a) -- (-45:0.4);
	\draw (a) -- (-135:0.4);
\end{tikzpicture}}
\newcommand{\dottickunit}[1]{%
\begin{tikzpicture}[dotpic,yshift=-1mm]
\node [#1] (0) at (0,0.25) {}; 
\draw (0)-- node[tick]{-} (0,-0.15);
\end{tikzpicture}}
\newcommand{\dottickcounit}[1]{%
\begin{tikzpicture}[dotpic,yshift=1mm]
\node [#1] (0) at (0,-0.25) {}; 
\draw (0)-- node[tick]{-} (0,0.15);
\end{tikzpicture}}
\newcommand{\dotonly}[1]{%
\begin{tikzpicture}[dotpic]
\node [#1] (0) at (0,0) {};
\end{tikzpicture}\,}
\newcommand{\dotthreestate}[1]{%
\begin{tikzpicture}[dotpic,yshift=2.5mm]
	\node [#1] (a) at (0,0) {};
	\draw (a) -- (0,-0.6);
	\draw [bend right] (a) to (-0.4,-0.6) (0.4,-0.6) to (a);
\end{tikzpicture}}

\newcommand{\tick}{%
\begin{tikzpicture}[dotpic]
	\node [style=none] (0) at (0,0.25) {};
	\node [style=none] (1) at (0,-0.25) {};
	\draw (0) to node[tick]{-} (1);
\end{tikzpicture}}

\newcommand{\lolli}{%
\begin{tikzpicture}[dotpic]
	\path [use as bounding box] (-0.25,-0.25) rectangle (0.25,0.5);
	\node [style=small dot] (0) at (0, 0.15) {};
	\node [style=none] (1) at (0, -0.25) {};
	\draw  (0) to (1.center);
	\draw [out=45, looseness=2.00, in=135, loop] (0) to ();
\end{tikzpicture}}

\newcommand{\cololli}{%
\begin{tikzpicture}[dotpic]
	\path [use as bounding box] (-0.25,-0.5) rectangle (0.25,0.5);
	\node [style=none] (0) at (0, 0.5) {};
	\node [style=small dot] (1) at (0, 0) {};
	\draw [out=-45, looseness=2.00, in=-135, loop] (1) to ();
	\draw  (1) to (0.center);
\end{tikzpicture}}

\newcommand{\blackdot}{\dotonly{small dot}}
\newcommand{\unit}{\dotunit{small dot}}
\newcommand{\counit}{\dotcounit{small dot}}
\newcommand{\mult}{\dotmult{small dot}}
\newcommand{\comult}{\dotcomult{small dot}}
\newcommand{\tickunit}{\dottickunit{small dot}}

\newcommand{\threestate}{\dotthreestate{small dot}}

\newcommand{\whitedot}{\dotonly{small white dot}}
\newcommand{\whiteunit}{\dotunit{small white dot}}
\newcommand{\whitecounit}{\dotcounit{small white dot}}
\newcommand{\whitemult}{\dotmult{small white dot}}
\newcommand{\whitecomult}{\dotcomult{small white dot}}

\newcommand{\whitetickcounit}{\dottickcounit{small white dot}}

\newcommand{\spider}[4][dot]{\node [#1] (#2) at (0,0) {};
\node [bn] (#2_d1) at (-1,1) {};
\node [bn] (#2_d2) at (-0.5,1) {};
\node [bn] (#2_dm) at (1,1) {};
\node [bn] (#2_c1) at (-1,-1) {};
\node [bn] (#2_c2) at (-0.5,-1) {};
\node [bn] (#2_cn) at (1,-1) {};

\node [anchor=west] at (#2_dm.east) {$#3$};
\node [anchor=west] at (#2_cn.east) {$#4$};
\node at (0.2,0.7) {\small{...}};
\node at (0.2,-0.7) {\small{...}};

\draw (#2)--(#2_d1) (#2)--(#2_d2) (#2)--(#2_dm);
\draw (#2)--(#2_c1) (#2)--(#2_c2) (#2)--(#2_cn);}

\newcommand{\circl}{\begin{tikzpicture}[dotpic]
	\node [circle,draw=black,inner sep=1pt] {\footnotesize\sf\phantom{$-$}};
\end{tikzpicture}}

\newcommand{\icircl}{\begin{tikzpicture}[dotpic]
	\node [circle,draw=black,inner sep=1pt] {\footnotesize\sf{}{$-$}};
\end{tikzpicture}}

\newcommand{\rtcircl}{\ensuremath{\sqrt{\begin{tikzpicture}[dotpic]
	\node [circle,draw=black,inner sep=1pt] {\tiny\sf\phantom{$-$}};
\end{tikzpicture}}}}
\newcommand{\rticircl}{\ensuremath{\sqrt{\begin{tikzpicture}[dotpic]
	\node [circle,draw=black,inner sep=1pt] {\tiny\bf\sf{}{$-$}};
\end{tikzpicture}}}}

\newcommand{\dcircl}{\begin{tikzpicture}[dotpic]
	\draw [use as bounding box,draw=none] (-0.15,-0.3) rectangle (0.15,0.3);
	\node [small dot] (0) {};
	\draw [uploop] (0) to ();
	\draw [downloop] (0) to ();
\end{tikzpicture}}

\newcommand{\ketGHZ}{\ket{\textit{GHZ}\,}}
\newcommand{\ketW}{\ket{\textit{W\,}}}

\def\bR{\begin{color}{red}} 
\def\bB{\begin{color}{blue}}
\def\bM{\begin{color}{magenta}}
\def\bC{\begin{color}{cyan}}
\def\bW{\begin{color}{white}}
\def\bBl{\begin{color}{black}}
\def\bG{\begin{color}{green}}
\def\bY{\begin{color}{yellow}}
\def\e{\end{color}}

\title{The compositional structure of\\ multipartite quantum entanglement} 
\titlerunning{The compositional structure of multipartite quantum entanglement}
\author{
Bob Coecke and Aleks Kissinger \thanks{This work is supported by EPSRC Advanced Research Fellowship EP/D072786/1, by a Clarendon Studentship, by US Office of Naval Research Grant N00014-09-1-0248 and by EU FP6 STREP QICS.  We thank Michael Hermann and Jamie Vicary for useful feedback.}  
}
\institute{
 Oxford University Computing Laboratory\\
Wolfson Building, Parks Road, Oxford OX1 3QD, UK\\
{\tt coecke / alek@comlab.ox.ac.uk}} 
\authorrunning{Bob Coecke and Aleks Kissinger}

\date{\today}
\def\makeheadbox{}

\begin{document}
\maketitle
\begin{abstract}
While multipartite quantum states constitute a (if not the) key resource for quantum computations and protocols, obtaining a high-level, structural understanding of entanglement involving arbitrarily many qubits is a long-standing open problem in quantum computer science. In this paper we expose the algebraic and graphical structure of the GHZ-state and the W-state, as well as a purely graphical distinction that characterises the behaviours of these states. In turn, this structure yields a compositional graphical model for expressing general multipartite states.

\hspace{3mm} 
We identify those states, named Frobenius states, which canonically induce an algebraic structure, namely the structure of a commutative Frobenius algebra (CFA). We show that all SLOCC-maximal tripartite qubit states are locally equivalent to Frobenius states. Those that are SLOCC-equivalent to the GHZ-state induce special commutative Frobenius algebras, while those that are SLOCC-equivalent to the W-state induce what we call anti-special commutative Frobenius algebras.  From the SLOCC-classification of tripartite qubit states follows a representation theorem for two dimensional CFAs.

\hspace{3mm} 
Together, a GHZ and a W Frobenius state form the primitives of a graphical calculus. This calculus is expressive enough to generate and reason about  arbitrary multipartite states, which are obtained by ``composing'' the GHZ- and W-states, giving rise to a rich graphical paradigm for general multipartite entanglement.
\end{abstract} 


\section{Introduction}

Spatially separated compound quantum systems exhibit correlations under measurement that cannot be explained by classical physics. Bipartite entangled states are used in protocols such as quantum teleportation, quantum key distribution, superdense coding, and entanglement swapping. 
The tripartite \em GHZ-state \em allows for a purely qualitative Bell-type argument demonstrating the non-locality of quantum mechanics \cite{GHZ}, a phenomenon which has recently been exploited to boost computational power \cite{AndersBrowne}.  In one-way quantum computing, multipartite \em  graph states \em which generalise GHZ-states constitute a resource for universal quantum computing \cite{one-way}. There are also many other applications of GHZ-states and graph states in the areas of fault-tolerance and  communication protocols \cite{Markham}. The tripartite \em W-state\em, which is qualitatively very different from the GHZ-state, supports a class of protocols implementing distributed leader-election \cite{Prakash}. The classification of tripartite qubit states even has applications in the study of extremal black holes in the STU supergravity theory, and certain canonical entanglement measures provide the Bekenstein-Hawking entropy \cite{Borsten1,Borsten2}. 

However, very little is known about the structure and behaviours of general multipartite quantum states, other than that the variety of possible behaviours is huge! For example, there is an infinite number of $4$-qubit states which are not inter-convertible by \emph{stochastic local} (quantum) \emph{operations} and \emph{classical communication} (SLOCC) \cite{Verstraete}. States that are not SLOCC-equivalent correspond to incomparable forms of ``distributed quantum-ness,'' so each will have distinct behaviours and applications, within and outside quantum computing.

For three qubits, SLOCC-classification is well understood:  there are only two non-degenerate SLOCC-classes \cite{DVC}, one that contains the GHZ-state and another one that contains the W-state:\vspace{-1mm}
\[
\ketGHZ=\ket{000}+\ket{111} 
\qquad\quad
\ketW=\ket{100}+\ket{010}+\ket{001}\,.\vspace{-1mm}
\]
This raises the question whether one can pinpoint in an elegant manner the essential structural difference between these classes, which may then, in turn, indicate behavioural differences. This is the primary goal of this paper. We provide algebraic characterisations of these two kinds of 3-partite states, and show that their difference is essentially of a \emph{topological} nature, specifically whether or not loops in their respective graphical representations disconnect the graph. Since we can interpret (dis)connectedness as the existence or non-existence of a flow of information, this topological distinction entails a behavioural one.

The guiding heuristic of this approach is that \em tripartite states are the same as algebraic operations\em. The most fundamental operations in nearly all branches of mathematics have two inputs and one output. Quantum protocols like gate teleportation \cite{Gottesman} which rely on \em map-state duality \em (or more general, the Choi-Jamiolkowski isomorphism \cite{Paulsen}) have taught us not to care much about distinguishing inputs and outputs, as there is little distinction between information flowing forward through an operation and flowing ``horizontally'' across an entangled state. Hence the \emph{total arity} (inputs + outputs) becomes a cogent feature. In this sense, tripartite states, regarded as operations with zero inputs and three outputs, are analogous to binary algebraic operations. We shall identify a class of states called \emph{Frobenius states} that give rise to particularly well-behaved algebraic structures called commutative Frobenius algebras (CFAs). These induced Frobenius algebras can then be manipulated graphically, using powerful techniques that take advantage of all of their symmetries.

We also show that GHZ-states and W-states play a foundational role in the composition of multipartite entangled states. The induced algebraic structures of these states can be used to construct \emph{arbitrary} multipartite states. GHZ-states and W-states will provide two kinds of primitive elements, which we can compose graphically to yield complex multipartite behaviours. This is similar to how arbitrary graph states can be formed by composing GHZ states, but the introduction of the W state brings a much greater variety of states and behaviours.

Compositionality is an important concept in modern computer science. From  elementary systems one can construct large, complex systems. Furthermore, one can reduce the (intractable) task of reasoning about the large system as a whole to the simpler task of reasoning about the interaction of its components. We apply this doctrine to reasoning about complex multipartite states as graphs of simpler states. In particular, the systems described in this paper are well-suited to manipulation with {\tt quantomatic}, a tool written by Dixon, Duncan and one of the authors \cite{quantomatic} for reasoning about interacting graphical structures.  This software combined with the results in this paper will enable one to automatically explore the vast space of multipartite entangled states, possibly guided by specification of an application.

While category theory is not a prerequisite for this paper, we do want to stress that the notion of a multipartite state, its induced algebra, and composition of these states and algebras already exists at the abstract level of \emph{symmetric monoidal categories} \cite{MacLane}.  In that sense, the results of this paper fit within the context of  the \em categorical quantum mechanics \em research program,  initiated by Abramsky and one of the authors in \cite{AC}, which aims to axiomatise quantum mechanical concepts in the language of symmetric monoidal categories. These categories are exactly the symbolic representations of the graphical language that we will use throughout this paper \cite{JoyalStreet,SelingerSurvey}, which traces back to Penrose's work in the 1970's \cite{Penrose}. Frobenius algebras, specifically in their categorical formulation \cite{CarboniWalters}, have become a key focus of categorical quantum mechanics by accounting for observables and corresponding classical data flows \cite{CPav,CPaqPav}, in particular due to their elegant normal forms \cite{Kock,Lack}. Complementary observables were axiomatised in terms of Frobenius algebras in \cite{CoeckeDuncan}, and this axiomatisation was used to solve an open problem in measurement based quantum computing in \cite{DuncanPerdrix2}, by relying on the graphical methods.

\paragraph{Structure of the paper}  In Section \ref{sec:graphical}, we recall graphical ``Penrose-style'' notation for compositions of tensors. In Section \ref{sec:Frob}, we introduce Frobenius algebras and their graphical language. We review SLOCC-classification of multipartite entangled states in Section \ref{sec:SLOCC}. In Section \ref{sec:frobenius-states} we define the notion of Frobenius state and illustrate its equivalence with commutative Frobenius algebra. Section \ref{sec:antispecial} defines special and anti-special commutative Frobenius algebras and Section \ref{sec:classification_states} establishes their correspondence with GHZ and W states, respectively.  Section \ref{sec:classification_CFA} classifies commutative Frobenius algebras on $\mathbb{C}^2$.  Section \ref{sec:interacting-states} defines a graphical interaction theory for the Frobenius algebras induced by related GHZ and W states called GHZ/W-pairs. It then goes on to detail how a GHZ/W-pair is universal with respect to state construction and draws comparisons with the inductive technique of Lamata et al \cite{Lamata}.
 
\par\medskip 
\paragraph{Note}  An extended abstract on part of the results presented in this paper was accepted for the 37th International Colloquium on Automata, Languages and Programming (ICALP) and is published in the conference proceedings thereof, volume 6199 of  Springer's Lecture Notes in Computer Science. 

\section{Graphical Notation for Tensors}\label{sec:graphical}

For Hilbert spaces $H_i$, we shall treat tensors as linear maps between the tensor products of spaces.

\[ f : H_1 \otimes H_2 \rightarrow H_3 \otimes H_4 \otimes H_5 \]

Such maps compose as usual. For $f : H_1 \otimes H_2 \rightarrow H_3$, $g : H_3 \rightarrow H_4$, we call the usual composition of maps $g \circ f : H_1 \otimes H_2 \rightarrow H_4$ the \emph{vertical composition}. The tensor product also extends to maps, so there is a map $f \otimes g : H_1 \otimes H_2 \otimes H_3 \rightarrow H_3 \otimes H_4$. We will call this \emph{horizontal composition}.

We use suggestive terms for these types of composition because the manipulation of multilinear maps is in a strong sense a 2-dimensional enterprise. One interpretation for this dimensionality is that the tensor product provides a spacial dimension, while the composition of arrows provides temporal, or causal dimension. The interplay of these two dimensions is represented by the interchange property of the tensor product.
\begin{equation}\label{eq:bifunc}
	(f_2 \otimes g_2) \circ (f_1 \otimes g_1) = (f_2 \circ f_1) \otimes (g_2 \circ g_1)
\end{equation}

One can interpret this equation by thinking of these four arrows occupying a piece of 2-dimensional space.

\begin{center}
\beginpgfgraphicnamed{2d_space}
\begin{tikzpicture}[scale=1.7]
	\node [dashed box] (f1) at (0,0) {$f_1$};
	\node [dashed box] (f2) at (0,-1) {$f_2$};
	\node [dashed box] (g1) at (1,0) {$g_1$};
	\node [dashed box] (g2) at (1,-1) {$g_2$};
\end{tikzpicture}}
\endpgfgraphicnamed
\end{center}

From this point of view, the bracketing in Eq \ref{eq:bifunc} is a piece of essentially meaningless syntax, which is required to make something that is 2-dimensional by nature expressible as a (1-dimensional) term. To address this issue, we shall introduce a Penrose-style graphical notation for multilinear maps, similar to that of circuit diagrams. Edges represent spaces and boxes represent multilinear maps. Normally, both edges and nodes are labeled, but here we shall consider only graphs where every edge represents the same fixed Hilbert space $H$, so we shall omit edge labels. Tensoring is done by juxtaposition and composition is performed by \emph{plugging}, or gluing the inputs of one graph to the outputs of another. The identity arrow is represented by a blank edge, and we omit the wire all together for the 1-dimensional Hilbert space $\mathbb C$, because $H \otimes \mathbb C \cong H \cong \mathbb C \otimes H$.

\begin{equation}\label{eqn:graphical-compose}
\beginpgfgraphicnamed{gluing_and_juxt}
\InputIfFileExists{gluing_and_juxt.tikz}{}{\input{./figures/gluing_and_juxt.tikz}}
\endpgfgraphicnamed
\end{equation}

Since we can express $H \otimes H$ as a pair of lines, we can, for instance, express a map $h : H \otimes H \rightarrow H \otimes H$ and compose in various ways with other maps.

\begin{equation}\label{eqn:mixed-arities}
\beginpgfgraphicnamed{mixed_arities}
\InputIfFileExists{mixed_arities.tikz}{}{\input{./figures/mixed_arities.tikz}}
\endpgfgraphicnamed
\end{equation}

Replacing some of the maps in Eq \ref{eq:bifunc} with identity maps, we can use the interchange identity to ``slide'' two linear maps past each other.
\begin{eqnarray*}
	(f \otimes 1_H) \circ (1_H \otimes g)
	  & = & (f \circ 1_H) \otimes (1_H \circ g) \\
	  & = & f \otimes g \\
	  & = & (1_H \circ f) \otimes (g \circ 1_H) \\
	  & = & (1_H \otimes g) \circ (f \otimes 1_H)
\end{eqnarray*}
We can express these equations graphically as follows:

\begin{equation}\label{eqn:graphical-bifunc}
\beginpgfgraphicnamed{bifunc}
\InputIfFileExists{bifunc.tikz}{}{\input{./figures/bifunc.tikz}}
\endpgfgraphicnamed
\end{equation}

$H \otimes H$ comes with a canonical swap isomorphism $\sigma : H \otimes H \rightarrow H \otimes H$:
\[ \sigma :: \ket\psi \otimes \ket\phi \mapsto \ket\phi \otimes \ket\psi \]
Graphically, we represent this as two crossed wires, and note for any $f \otimes g$,
\begin{equation}\label{eqn:graphical-symmetry}
\beginpgfgraphicnamed{symmetry_natural}
\InputIfFileExists{symmetry_natural.tikz}{}{\input{./figures/symmetry_natural.tikz}}
\endpgfgraphicnamed
\end{equation}

Its easy to see how states are a special case of this notation when we think of a state $\ket\psi \in H$ as a linear map from $\mathbb C$ to $H$.
\[ \ket\psi :: 1 \mapsto \ket\psi \]
Since $\mathbb C$ is written as ``no wire,'' we omit the line when a map goes to or from the complex numbers:

\begin{center}
\beginpgfgraphicnamed{bras_and_kets}
\InputIfFileExists{bras_and_kets.tikz}{}{\input{./figures/bras_and_kets.tikz}}
\endpgfgraphicnamed
\end{center}

We can compose a bra and a ket as usual to obtain the inner product. This is a linear map from $\mathbb C$ to $\mathbb C$, i.e. just a complex number.

\begin{center}
\beginpgfgraphicnamed{inner_product}
\begin{tikzpicture}
	\node [style=pt] (0) at (0.5, 0.25) {$\psi$};
	\node [style=none, anchor=east] (1) at (0, 0) {$\braket{\phi}{\psi} =$};
	\node [style=copt] (2) at (0.5, -0.25) {$\phi$};
	\draw  (0) to (2);
\end{tikzpicture}}
\endpgfgraphicnamed
\end{center}

We shall state a concrete version of a soundness theorem for the graphical language, due to Joyal and Street \cite{JoyalStreet}.

\begin{theorem}[Soundness of the graphical calculus]\label{thm:soundness-graphical}
    Two maps consisting of arbitrary compositions and tensor products of smaller maps $f : \bigotimes A_i \rightarrow \bigotimes B_j$ and swap maps $\sigma_{A,B}$ are equal if their graphical representations are equal.
\end{theorem}

Restricting to certain kinds of linear maps, namely ``caps'' and ``cups'' which we shall introduce briefly in \ref{sec:Frob}, Selinger proved the ``only if'' part of the above theorem in \cite{SelingerCompleteness}.

\section{Commutative Frobenius algebras}\label{sec:Frob}

To fix notations, we recall the usual notion of unital algebra. 

Consider a vector space $A$ equipped with a multiplication map $(- \cdot -) : A \times A \rightarrow A$. We say that $(A, \cdot)$ is a \emph{unital algebra} if

\begin{itemize}
	\item $(- \cdot -)$ is bilinear,
	\item $(\ket u \cdot \ket v) \cdot \ket w =
	        \ket u \cdot (\ket v \cdot \ket w)$ for all
	        $\ket u, \ket v, \ket w \in A$, and
	\item there exists $\ket \eta \in A$ such that $\ket u \cdot \ket\eta = \ket\eta \cdot \ket u = \ket u$ for all $\ket u \in A$.
\end{itemize}

For our purposes, we shall assume $A$ is always a finite dimensional complex Hilbert space $H$. Since $(- \cdot -)$ is bilinear, there exists a unique $\mu : H \otimes H \rightarrow H$ such that

\[ \mu(\ket u \otimes \ket v) = \ket u \cdot \ket v \]

Taking this to define multiplication, we obtain the following definition.

\begin{definition}\label{def:assoc-alg}\em 
	A \emph{unital algebra} $(H,\mu,\eta)$ is a vector space $H$ with maps $\mu : H \otimes H \rightarrow H$, $\eta : \mathbb C \rightarrow H$ such that $\mu(1 \otimes \mu) = \mu(\mu \otimes 1)$ and $\mu(1 \otimes \eta) = \mu(\eta \otimes 1) = 1$.
\end{definition}

We can also form a counital coalgebra. This is a unital algebra on the dual space $H^*$.

\[ (H^*, \delta^* : H^* \otimes H^* \rightarrow H^*, \epsilon^* \in H^*) \]

To clarify, we make a direct definition in terms of $H$, rather than $H^*$.

\begin{definition}\label{def:coassoc-alg}\em
	A \emph{counital coalgebra $(H, \delta, \epsilon)$} is a vector space $H$ with a map $\delta : H \rightarrow H \otimes H$ called the comultiplication and a map $\epsilon : H \rightarrow \mathbb C$ called the counit such that $(1 \otimes \delta)\delta = (\delta \otimes 1)\delta$ and $(\epsilon \otimes 1)\delta = (1 \otimes \epsilon)\delta = 1$.
\end{definition}

Let $\sigma_{A,B}$ be the swap map. Then, an algebra (resp. coalgebra) is \emph{commutative} (resp. \emph{cocommutative}) iff $\mu = \mu\sigma_{H,H}$ (resp. $\delta = \sigma_{H,H}\delta$).

There are many situations where we can choose an algebra $(H, \mu, \eta)$ and a coalgebra $(H, \delta, \epsilon)$ that interact well together. The case of interest here is that of \emph{Frobenius algebras}.

\begin{definition}\label{def:concrete-Frobenius}\em 
	A \emph{Frobenius algebra} $\mathcal F$ is a vector space $H$ with maps $\mu, \eta, \delta, \epsilon$ such that
	\begin{enumerate}
		\item $(H, \mu, \eta)$ is a unital algebra,
		\item $(H, \delta, \epsilon)$ is a counital coalgebra, and
		\item $(\mu \otimes 1)(1 \otimes \delta) =
					(1 \otimes \mu)(\delta \otimes 1) = \delta \mu$ (Frobenius identity).
	\end{enumerate}
\end{definition}

Depicting $\mu$, $\delta$, $\eta$, $\epsilon$ respectively as $\mult$, 
$\comult$, $\unit$ and $\counit$, we can rephrase this definition using graphical identities.

\begin{definition}\label{def:graphical-frob}\em
    A \emph{Frobenius algebra} $\blackdot$ is a vector space $H$ with maps $\mult, \unit, \comult, \counit$ such that the following equations hold:
    
    \bigskip

\beginpgfgraphicnamed{frobenius_ids}
\InputIfFileExists{frobenius_ids.tikz}{}{\input{./figures/frobenius_ids.tikz}}
\endpgfgraphicnamed
\end{definition}

By convention, we shall refer to Frobenius algebras either as script letters $\mathcal F$, $\mathcal G$ or by the colour of their dots, e.g. $\blackdot$, $\whitedot$. Also, note that we shall refer to a Frobenius algebra without a unit or counit as a \emph{Frobenius semi-algebra}.

\begin{example}
	Let $M$ be the vector space of $n \times n$ matrices. Take $\mu$ to be matrix multiplication, which is associative and bilinear. Let $\eta$ be the $n\times n$ identity matrix, and let $\epsilon : M \rightarrow k$ be the trace functional. This data induces a unique map $\delta$ such that $(M, \mu, \eta, \delta, \epsilon)$ is a Frobenius algebra.
\end{example}

\begin{definition}\label{def:commutative-frobenius-algebra}\em 
	A Frobenius algebra $(H, \mu, \eta, \delta, \epsilon)$ is called \emph{commutative} when the unital algebra $(H, \mu, \eta)$ is commutative and the counital coalgebra $(H, \delta, \epsilon)$ is cocommutative, or as graphical identities:
	    
    \bigskip

\beginpgfgraphicnamed{com}
\InputIfFileExists{com.tikz}{}{\input{./figures/com.tikz}}
\endpgfgraphicnamed
\end{definition}

Commutative Frobenius algebras, or CFA's, will be our primary focus for the remainder of this article.

\begin{example}[CFAs for an orthonormal  basis]\label{ex:basisCFA}
Given an orthonormal basis 
\[
{\{\ket{0}, \ldots, \ket{d-1}\}}
\]
of  $\mathbb{C}^d$, the linear maps
\[
\mu= \sum_i \ket{i}\bra{ii}:\mathbb{C}^d\otimes \mathbb{C}^d\to \mathbb{C}^d
\qquad
\eta= \sum_i \ket{i}:\mathbb{C}\to \mathbb{C}^d
\]
\[
\delta= \sum_i \ket{ii}\bra{i}:\mathbb{C}^d\to \mathbb{C}^d\otimes \mathbb{C}^d
\qquad
\epsilon= \sum_i \bra{i}: \mathbb{C}^d\to \mathbb{C}
\]
define a CFA on $\mathbb{C}^d$. For any finite orthonormal basis, this defines a Frobenius algebra whose comultiplication copies the basis vectors and whose counit uniformly deletes them (i.e. sends them all to the scalar 1). Note also that $\mu = \delta^\dagger$ and $\eta = \epsilon^\dagger$.
\end{example}

\begin{definition}\label{def:f-graph}\em
	For a CFA $\mathcal F = (H, \mu, \eta, \delta, \epsilon)$, an \em $\mathcal F$-graph \em is a map obtained from $1_H$, $\sigma_{H,H}$, $\mu$, $\eta$, $\delta$, and $\epsilon$, combined with composition and the tensor product. An $\mathcal F$-graph is said to be \em connected \em precisely when its graphical representation is connected.
\end{definition}

There is a well-known result about commutative Frobenius algebras, namely that any connected $\mathcal F$-graph is uniquely determined by its number of inputs, outputs, and its number of loops, where  by the \em number of loops \em we mean the maximum number of edges one can remove without disconnecting the graph. This result was rigorously proved by Kock, using the language of smooth manifolds and cobordisms~\cite{Kock}. In this formulation, the number of loops is just the genus of the associated manifold.

This makes CFAs highly topological, in that $\mathcal F$-graphs are invariant under deformations that respect the number of loops. Using these deformations, we can find a normal form that moves all loops to the middle, all multiplications to the top, and all comultiplications to the bottom.

\begin{theorem}\label{thm:cfa-nf}
Any connected $\mathcal F$-graph can always be put in the following normal form:
\begin{equation}\label{eq:normalform}
\beginpgfgraphicnamed{cfa_full_nf}
\InputIfFileExists{cfa_full_nf.tikz}{}{\input{./figures/cfa_full_nf.tikz}}
\endpgfgraphicnamed
\end{equation}
\end{theorem}

\subsection{Spider notation}

When there are no loops, an $\mathcal F$-graph is actually an $\mathcal F$-tree. A connected $\mathcal F$-tree is called a \emph{spider}. We denote the unique $\mathcal F$-tree with $n$ in-edges and $m$ out-edges as follows.

\[
S^n_m : \overbrace{H \otimes \ldots \otimes H}^n
\rightarrow
\underbrace{H \otimes \ldots \otimes H}_m \]

Graphically, we express spiders as vertices with any number of incident edges.

\begin{center}
\beginpgfgraphicnamed{new_spider_def}
\InputIfFileExists{new_spider_def.tikz}{}{\input{./figures/new_spider_def.tikz}}
\endpgfgraphicnamed
\end{center}

In the event that a spider has zero inputs or zero outs, we ``cap off'' the end with $\unit$ or $\counit$.

\begin{center}
	$S^0_m := S^1_m \circ \unit$
	\qquad
	$S^n_0 := \counit \circ S^n_1$
	\qquad
	$S^1_1 := 1_A$
\end{center}

This notation will become useful for denoting maps of arbitrary arity. Note that $S^0_n$ is an $n$-partite state and $S^m_0$ is an $m$ partite effect.

\begin{example}[spiders for a basis]
	Let $\mathcal G$ be the commutative Frobenius algebra for the orthonormal basis $\{ \ket{0}, \ket{1} \}$, defined as in Ex \ref{ex:basisCFA}. Spiders for $\mathcal G$ are
\[
S^n_m = \sum_{i=0,1} \underbrace{\ket{i\ldots i}}_m\overbrace{\bra{i\ldots i}}^n:
\overbrace{\mathbb{C}^2\otimes\ldots\otimes \mathbb{C}^2}^n\to 
\underbrace{\mathbb{C}^2\otimes\ldots\otimes \mathbb{C}^2}_m.
\] 
In particular, spiders without inputs, 
\begin{center}
\beginpgfgraphicnamed{spiders_without_inputs}
\InputIfFileExists{spiders_without_inputs.tikz}{}{\input{./figures/spiders_without_inputs.tikz}}
\endpgfgraphicnamed
\end{center}
are exactly the $m$ party generalisations of GHZ states:
\[
\ket{\textit{GHZ}_m} = \ket{0\ldots0} + \ket{1\ldots1}.
\]
Spiders without outputs are the corresponding effects, e.g. $S^2_0=\bra{00}+\bra{11}$.
\end{example}

\subsection{Caps, cups, and the partial trace}

One of the useful side effects of defining a Frobenius algebra is that any Frobenius algebra automatically fixes an isomorphism with the dual space. Consider the following two maps:

\begin{center}
\beginpgfgraphicnamed{cap_and_cup}
\InputIfFileExists{cap_and_cup.tikz}{}{\input{./figures/cap_and_cup.tikz}}
\endpgfgraphicnamed
\end{center}

These induce maps between $H$ and its dual space $H^*$ as follows.

\begin{center}
\beginpgfgraphicnamed{induced_isos}
\InputIfFileExists{induced_isos.tikz}{}{\input{./figures/induced_isos.tikz}}
\endpgfgraphicnamed
\end{center}

\begin{proposition}\label{pro:induced-iso}
	The maps $\widehat{S^2_0}$ and $\widehat{S^0_2}$ define an isomorphism of Hilbert spaces $H \cong H^*$.
\end{proposition}

\begin{proof}
	This follows as an easy consequence of Thm \ref{thm:cfa-nf}.

	\begin{center}
\beginpgfgraphicnamed{induced_iso_pf}
\InputIfFileExists{induced_iso_pf.tikz}{}{\input{./figures/induced_iso_pf.tikz}}
\endpgfgraphicnamed
	\end{center}
	
	So $\widehat{S^0_2} \circ \widehat{S^2_0} = 1_H$. The fact that $\widehat{S^2_0} \circ \widehat{S^0_2} = 1_{H^*}$ can be shown similarly.
\qed
\end{proof}

In fact, a unital algebra $(H, \mu, \eta)$ with some linear functional $\epsilon : H \rightarrow \mathbb C$ such that $\widehat{\mu \circ \epsilon}$ ($= \widehat{S^2_0}$) is an isomorphism from $H$ to $H^*$ is an equivalent (and much older) definition of a Frobenius algebra. The comultiplication in this case is completely determined by the multiplication and this induced isomorphism.

Such caps and cups are a special case of a category-theoretic construction known as a \emph{compact structure} \cite{KellyLaplaza}. Compact structures provide an abstract notion of what it means for a vector space to have a dual space. 

For a commutative Frobenius algebra $\blackdot = (H, \mult, \comult, \unit, \counit)$, we define the following operator.

\begin{center}
\beginpgfgraphicnamed{partial_trace}
\InputIfFileExists{partial_trace.tikz}{}{\input{./figures/partial_trace.tikz}}
\endpgfgraphicnamed
\end{center}

We first note that this operator doesn't depend on the choice of CFA.

\begin{proposition}\label{pro:trace-any-frob}
	For any two commutative Frobenius algebras $\blackdot = (H, \mult, \comult, \unit, \counit)$ and $\whitedot = (H, \whitemult, \whitecomult, \whiteunit, \whitecounit)$, $\textrm{\rm Tr}^{\tikz{\node[dot]{};}}_H = \textrm{\rm Tr}^{\tikz{\node[white dot]{};}}_H$.
\end{proposition}

\begin{proof}
	We prove this by applying Thm \ref{thm:cfa-nf}.
	\begin{center}
\beginpgfgraphicnamed{partial_trace_canonical_pf}
\InputIfFileExists{partial_trace_canonical_pf.tikz}{}{\input{./figures/partial_trace_canonical_pf.tikz}}
\endpgfgraphicnamed \qed
	\end{center}
\end{proof}

\begin{proposition}\label{pro:partial-trace}
	For any CFA, $\textrm{\rm Tr}^{\tikz{\node[dot]{};}}_H(M)$ performs the partial trace of $M$.
\end{proposition}

\begin{proof}
	From Prop \ref{pro:trace-any-frob}, we know we can choose any CFA to define $\textrm{\rm Tr}^{\tikz{\node[dot]{};}}_H(M)$ for $M \in H' \otimes H$. Let it be the CFA given in example \ref{ex:basisCFA}. In this case, the cap is $\sum \ket{kk}$ and the cup is $\sum \bra{kk}$. The claim then follows straightforwardly.\qed
	
\end{proof}

We note in particular that

\begin{center}
\beginpgfgraphicnamed{tr_id}
\InputIfFileExists{tr_id.tikz}{}{\input{./figures/tr_id.tikz}}
\endpgfgraphicnamed
\end{center}

As the circle is always the dimension of $H$ for any CFA over $H$, we write this scalar simply as $\circl$. We also note that the trace operator allows us to interpret self-loops in $\mathcal F$-graphs unambiguously.

\begin{center}
\beginpgfgraphicnamed{trace_selfloop}
\InputIfFileExists{trace_selfloop.tikz}{}{\input{./figures/trace_selfloop.tikz}}
\endpgfgraphicnamed
\end{center}

%
%
%

\section{SLOCC Classification}\label{sec:SLOCC}

We recall some facts about SLOCC-classification of multipartite states, since it is this classification which singles out the GHZ-state and the W-state as canonical.

\begin{definition}\label{def:locc}\em
	A state $\ket\Psi \in H_1 \otimes \ldots \otimes H_N$ can be converted into a state $\ket\Phi$ using \textit{Stochastic Local Operations and Classical Communication (SLOCC)} precisely when there exists an $N$ party protocol that succeeds with non-zero probability at turning $\ket\Psi$ into $\ket\Phi$, where each party $p_i$ has access to $H_i$, and can:
	\begin{itemize}
		\item apply any number of local (quantum) operations
		      $O : H_i \rightarrow H_i$
		\item perform any amount of classical communication
		      with the other parties
	\end{itemize}	
	In such a case, we write $\ket\Phi \preceq \ket\Psi$.
\end{definition}

Note that these operations $O$ could be implemented as generalised measurements, i.e.~those arising when measuring the system together with an ancillary system by means of a projective measurement, as well unitary operations applied to extended systems. 
The relation $\preceq$ forms a preorder, and the induced equivalence relation, called \em SLOCC-equivalence\em, is denoted $\sim$. 
The set of all states SLOCC-equivalent to a given state $\ket\Psi$ is 
the \emph{SLOCC class}, of $\ket\Psi$.

\begin{theorem}[\cite{DVC}]\label{thm:slocc-ilo}
	Two states $\ket\Psi, \ket\Phi \in H_1 \otimes \ldots H_N$ are SLOCC-equivalent iff there exist invertible linear maps $L_i: H_i \rightarrow H_i$ such that
	\[
	\ket\Psi = (L_1 \otimes  \ldots \otimes L_N)\ket\Phi\,.
	\]
\end{theorem}

We say a state $\ket\Psi$ is \emph{SLOCC-maximal} if it is maximal with respect to $\preceq$. I.e.

\[ \ket\Psi \preceq \ket\Phi \implies \ket\Phi \sim \ket\Psi \]

For two qubits there are two SLOCC classes:
\begin{center}
\beginpgfgraphicnamed{two-qubit-slocc-classes} 
	\begin{tikzpicture}[dotpic]
		\node [box vertex] (0) at (-1.00, 1.00) {$\ket{Bell}$};
		\node [box vertex] (1) at (-1.00, -1.00) {$\ket\psi \otimes \ket\phi$};
		\draw  (0) to node [auto] {$\preceq$} (1);
	\end{tikzpicture}
\endpgfgraphicnamed 
\end{center}

Where $\ket{Bell}$ is the usual Bell state $\ket{00}+\ket{11}$. Any entangled state on qubits can be obtained from the Bell state using local invertible maps. Any local rank-one map (e.g. the outcome of a projective measurement) converts the Bell state into a product state.

For three qubits, there are still only four SLOCC classes, up to a permutation of qubits, but $\preceq$ is no longer a total order. These classes are:

\begin{center}
\beginpgfgraphicnamed{three_qubit_slocc_classes} 
	\begin{tikzpicture}[dotpic]
		\node [box vertex] (0) at (-3.00, 1.00) {$\ketGHZ$};
		\node [box vertex] (1) at (1.00, 1.00) {$\ketW$};
		\node [box vertex] (3) at (-1.00, -1.00) {$\ket{Bell} \otimes \ket\psi$};
		\node [box vertex] (5) at (-1.00, -3.00) {$\ket\psi \otimes \ket\phi \otimes \ket\xi$};
		\draw  (1) to node [auto] {} (3);
		\draw  (0) to node [auto,swap] {} (3);
		\draw  (3) to node [auto] {$\preceq$} (5);
		\node [style=none] (0) at (-3.00, 0.00) {$\preceq$};
		\node [style=none] (1) at (1.00, 0.00) {$\preceq$};
	\end{tikzpicture}
\endpgfgraphicnamed 
\end{center}

Note that there is not a unique maximally entangled SLOCC-class, but two classes, represented by the GHZ and W states.

Beyond three qubits, SLOCC classification becomes much more difficult. This is because there are infinitely many distinct SLOCC classes when $N\geq 4$ \cite{DVC}. To obtain finite classification results many authors have expanded to talk about SLOCC \emph{super-classes}, or families of SLOCC classes parameterised by one or more continuous variables.  
However, such classifications depend highly on the mathematical strategy by which states are reduced to canonical forms, rather than on any intrinsic qualities of the states themselves. An example of such a classification in terms of SLOCC super-classes  is in Section \ref{sec:generating-states}.

\subsection{SLOCC-maximal entangled states}\label{sec:maximal-entanglement}

We shall now focus on the states that are maximally entangled with respect to SLOCC. This is a good starting point if we consider states as computational resources. For any state $\ket\Psi$ we can always find some maximally entangled state and a SLOCC protocol that produces $\ket\Psi$. As we intend to reason about such states graphically, we shall reformulate SLOCC-maximality in graphical terms.

\begin{proposition}\label{pro:bipartite-maximal}
	A bipartite state $\ket\Psi \in H \otimes H$ is SLOCC-maximal iff there exists an effect $\bra\Phi : H \otimes H \rightarrow \mathbb C$ such that
	\begin{equation}\label{eq:bipartite-maximal}
\beginpgfgraphicnamed{maximal_bipartite}
\InputIfFileExists{maximal_bipartite.tikz}{}{\input{./figures/maximal_bipartite.tikz}}
\endpgfgraphicnamed
	\end{equation}
\end{proposition}

The proof of this result follows straightforwardly by map-state duality.


\section{Frobenius States}\label{sec:frobenius-states}

As we've alluded to in previous sections, we shall look at the connection between tripartite entangled states and algebraic structures. We shall highlight two sufficient conditions for a tripartite state to generate structures that (a) satisfy a unit law and (b) are associative, commutative, and Frobenius. In a sense that will soon become clear, condition (a) implies that a tripartite state be highly entangled and condition (b) implies that the state be highly symmetric.

In this section, we shall consider the kinds of tripartite states $\ket\Psi$ that are not only highly entangled and symmetric, but are members of a family of $N$-partite states that can be constructed from $\ket\Psi$ which inherit this state's ``algebraic'' properties. To this end, we shall introduce the notions of \emph{strong SLOCC-maximality} and \emph{strong symmetry}.

\begin{definition}\label{def:tripartite-maximal}\em 
	A tripartite state $\ket\Psi \in H \otimes H \otimes H$ is said to be \emph{strongly SLOCC-maximal} if there exist effects $\bra{\xi_i}$ such that the following three states are all SLOCC-maximal.
	\begin{equation}\label{eq:strong-maximality}
\beginpgfgraphicnamed{maximal_tripartite}
\InputIfFileExists{maximal_tripartite.tikz}{}{\input{./figures/maximal_tripartite.tikz}}
\endpgfgraphicnamed
	\end{equation}
\end{definition}

\begin{theorem}\label{thm:strong-maximality}
    Strong SLOCC-maximality strictly implies SLOCC-maximality.
\end{theorem}

\begin{proof}
    Suppose $\ket\Psi$ were not SLOCC-maximal. Then there exists $\ket{\Psi'}$ such that
    \[ \ket\Psi = (S_1 \otimes S_2 \otimes S_3)\ket{\Psi'} \]
    for at least one $S_i$ a singular map. Therefore at least two of the three equations given in (\ref{eq:strong-maximality}) could not hold. So, strong SLOCC-maximality implies SLOCC-maximality.
    
    Conversely, suppose SLOCC-maximality implies strong SLOCC-maximality. We observe that any map $A \rightarrow A \otimes A$ can be written in this form:
    \begin{center}
\beginpgfgraphicnamed{state_comult}
\InputIfFileExists{state_comult.tikz}{}{\input{./figures/state_comult.tikz}}
\endpgfgraphicnamed
    \end{center}
    for $\bra\Phi = \sum \bra{ii}$. In particular, a comultiplication $\delta$ from a coalgebra $(A, \delta, \epsilon)$ could be expressed as above. In such a case, the unit law forces $\ket\Psi$ to be SLOCC-maximal. Supposing this implies strong SLOCC-maximality, then there exists $\bra\xi$ such that the following is SLOCC-maximal:
    \[ (1 \otimes 1 \otimes \bra\xi)\ket\Psi \]
    Let $\eta = (1 \otimes \bra\xi)\ket\Phi$. It is then straightforward to show that $(\delta, \eta)$ extends to a Frobenius algebra. Thus \emph{every} coalgebra extends to a Frobenius algebra, which is a contradiction.
\qed
\end{proof}

In the case where $\ket\Psi$ is symmetric, we can simplify this condition.

\begin{proposition}\label{pro:symmetric-maximal}
	A symmetric tripartite state $\ket\Psi$ is strongly SLOCC-maximal iff
	\begin{equation}\label{eq:symmetric-maximal}
\beginpgfgraphicnamed{maximal_symmetric}
\InputIfFileExists{maximal_symmetric.tikz}{}{\input{./figures/maximal_symmetric.tikz}}
\endpgfgraphicnamed
	\end{equation}
\end{proposition}

The GHZ state satisfies Eq \ref{eq:symmetric-maximal} when we fix $\bra\xi = \bra{+}$ and $\bra\Phi = \bra{\textit{Bell}\,}$. The W state satisfies this equation when we let $\bra\xi = \bra{0}$ and $\bra\Phi = \bra{\textit{EPR}\,} := \bra{01} + \bra{10}$.

As these are the only two maximally entangled tripartite states for qubits, it is worthwhile to expound upon their graphical properties. These states are both symmetric, and furthermore, have natural $N$-qubit symmetric analogues.
\begin{eqnarray}
	\ket{GHZ_N} & := & \ket{00\ldots0} + \ket{11\ldots1} \\
	\ket{W_N} & := & \ket{10\ldots0} + \ket{01\ldots0} +
	                 \ldots + \ket{0\ldots01}
\end{eqnarray}

Not only do they have such $N$-partite versions, they come with a recipe for inductively constructing them. That is, for both of these states, there is a bipartite effect $\bra\Phi$ that can be used to ``glue'' a tripartite state on to an $N$-partite state by projecting out a pair of qubits to make an $(N+1)$-partite symmetric state.
\begin{eqnarray*}
	\ket{GHZ_{N+1}} & = &
	  (1 \otimes \bra{\textit{Bell }} \otimes 1)
	  (\ket{GHZ_N} \otimes \ket{GHZ_3}) \\
	\ket{W_{N+1}} & = &
	  (1 \otimes \bra{\textit{EPR }} \otimes 1)
	  (\ket{W_N} \otimes \ket{W_3})
\end{eqnarray*}

To inductively build a symmetric $N$ party state, it suffices that the following condition hold.

\begin{definition}\label{def:inductively-symmetric}\em 
	A symmetric state is said to be \emph{strongly symmetric} if there exists some bipartite effect $\bra\Phi$ such that
	\begin{equation}\label{eq:inductively-symmetric}
\beginpgfgraphicnamed{inductively_symmetric}
\InputIfFileExists{inductively_symmetric.tikz}{}{\input{./figures/inductively_symmetric.tikz}}
\endpgfgraphicnamed
	\end{equation}
\end{definition}

We can now make a general definition for the types of states we consider to be ``highly algebraic'' in character.


\begin{definition}\label{def:Frobenius-state}\em 
	A symmetric tripartite state $\ket\Psi$ is said to be a \emph{Frobenius state} if there exist effects $\bra\Phi$, $\bra\xi$ such that Eqs \ref{eq:symmetric-maximal} and \ref{eq:inductively-symmetric} hold.
\end{definition}

Note that $\bra\Phi$ satisfying Eqs \ref{eq:symmetric-maximal} and \ref{eq:inductively-symmetric} must be the \emph{same effect}. This is a stronger condition than stating that these equations respectively hold for some effects $\bra\Phi$ and $\bra{\Phi'}$. This will be crucial to the construction to follow.

\begin{theorem}[Algebras as states]\label{thm:frob-state-from-cfa}
    For any commutative Frobenius algebra $\blackdot = (H, \mult, \unit, \comult, \counit)$, the following is a Frobenius state, with its two associated effects:
    \begin{center}
\beginpgfgraphicnamed{frob_state_from_cfa}
\InputIfFileExists{frob_state_from_cfa.tikz}{}{\input{./figures/frob_state_from_cfa.tikz}}
\endpgfgraphicnamed
    \end{center}
\end{theorem}

The Frobenius state conditions from Def \ref{def:Frobenius-state} hold as a consequence of Thm \ref{thm:cfa-nf}. Also, from any Frobenius state we can construct the associated commutative Frobenius algebra.

\begin{theorem}[States as algebras]\label{thm:frob-state-cfa}
    For any Frobenius state $\ket\Psi$, there exist effects $\bra\Phi$, $\bra\xi$ such that the following is a commutative Frobenius algebra:
    
    \begin{center}
\beginpgfgraphicnamed{cfa_from_frob_state}
\InputIfFileExists{cfa_from_frob_state.tikz}{}{\input{./figures/cfa_from_frob_state.tikz}}
\endpgfgraphicnamed 
    \end{center}
\end{theorem}

As a result, we can refer to commutative Frobenius algebras either as the usual maps $(\mu, \eta, \delta, \epsilon)$, or as a triple $(\ket\Psi, \bra\Phi, \bra\xi)$ consisting of a Frobenius state and its two associated effects. Also, note that for a given state $\ket\Psi$, there could be multiple induced commutative Frobenius algebras based upon the choice of $\bra\xi$. However, once $\bra\xi$ is fixed, $\bra\Phi$ is completely determined by Eq \ref{eq:symmetric-maximal}. This is analogous to the situation with Frobenius algebras where the maps $\mu$ and $\epsilon$ completely determine the other two.  The following table compares the previously stated definition of a CFA to the one in terms of tripartite  states:
\begin{center}
\begin{tabular}{|c|c|c|} 
\hline
data & commutative Frobenius semi-algebra & unitality  \\

\hline

\raisebox{-3.6mm}{$%
\beginpgfgraphicnamed{CFA_d}
\InputIfFileExists{CFA_d.tikz}{}{\input{./figures/CFA_d.tikz}}
\endpgfgraphicnamed$} & $%
\beginpgfgraphicnamed{CFA}
\InputIfFileExists{CFA.tikz}{}{\input{./figures/CFA.tikz}}
\endpgfgraphicnamed$ & \raisebox{2mm}{$%
\beginpgfgraphicnamed{CFA_u}
\InputIfFileExists{CFA_u.tikz}{}{\input{./figures/CFA_u.tikz}}
\endpgfgraphicnamed$} \\

\hline

$%
\beginpgfgraphicnamed{CFA_state_d}
\InputIfFileExists{CFA_state_d.tikz}{}{\input{./figures/CFA_state_d.tikz}}
\endpgfgraphicnamed$ & $%
\beginpgfgraphicnamed{CFA_state}
\InputIfFileExists{CFA_state.tikz}{}{\input{./figures/CFA_state.tikz}}
\endpgfgraphicnamed$ &  \raisebox{1mm}{$%
\beginpgfgraphicnamed{CFA_state_u}
\InputIfFileExists{CFA_state_u.tikz}{}{\input{./figures/CFA_state_u.tikz}}
\endpgfgraphicnamed$} \\

 \hline
\end{tabular}
\end{center}
We now show the induced Frobenius algebras for our two motivating examples of Frobenius states.

\begin{example}\label{ex:concreteGHZ}
    For the Frobenius state $\ketGHZ$, fixing $\bra\xi := \sqrt{2}\bra{+}$ induces the following CFA, which we shall refer to as $\mathcal G$.
    \begin{equation}\label{GHZ-SCFA}
    	\begin{split}
    		\whitemult & = \ket{0}\bra{00} + \ket{1}\bra{11} \qquad\qquad
    		\whiteunit = \sqrt{2}\, \ket{+} := \ket{0}+\ket{1} \\
    		\whitecomult & = \ket{00}\bra{0} + \ket{11}\bra{1} \qquad\qquad
    		\whitecounit = \sqrt{2} \bra{+} := \bra{0}+\bra{1}
    	\end{split}\vspace{-1.5mm}
    \end{equation}
\end{example}
\begin{example}\label{ex:concreteW}
    For the Frobenius state $\ketW$, fixing $\bra\xi := \bra{0}$ induces the following CFA, called $\mathcal W$.
    \begin{equation}\label{W-ACFA}
    	\begin{split}
    		\mult & = \ket{1}\bra{11} + \ket{0}\bra{01} + \ket{0}\bra{10} 
    		\qquad\qquad\qquad
    		\unit = \ket 1\qquad\qquad \\
    		\comult & = \ket{00}\bra{0} + \ket{01}\bra{1} + \ket{10}\bra{1} 
    		\qquad\qquad\qquad
    		\counit = \bra 0\qquad\qquad
    	\end{split}
    \end{equation}    
\end{example}

\section{Special and anti-special commutative Frobenius algebras}\label{sec:antispecial} 

The normal form given in Thm \ref{thm:cfa-nf} suggests that graphs of Frobenius algebras could contain many loops. We shall focus here on two cases of CFA's. The first is where the loops vanish, and the second is where they propagate outwards, disconnecting the entire graph.

\begin{definition}\label{def:scfa}\em 
A \emph{special commutative Frobenius algebra} (SCFA) is a commutative Frobenius algebra where
\begin{center}
\beginpgfgraphicnamed{special_cfa}
\begin{tikzpicture}[dotpic,yshift=5mm]
	\node [dot] (a) at (0,0) {};
	\node [dot] (b) at (0,-1) {};
	\draw [bend left] (a) to (b);
	\draw [bend right] (a) to (b);
	\draw (0,0.5) to (a) (b) to (0,-1.5);
\end{tikzpicture} $=$ \
\begin{tikzpicture}[dotpic]
	\draw (0,1) -- (0,-1);
\end{tikzpicture}}
\endpgfgraphicnamed
\end{center}
\end{definition}

\begin{theorem}\label{thm:scfa-spider}
	For a special commutative Frobenius algebra $\mathcal S$, any connected $\mathcal S$-graph with $n$ inputs and $m$ outputs is equal to the spider $S^n_m$.
\end{theorem}

\begin{proof}
	From Thm \ref{thm:cfa-nf}, we can put any connected $\mathcal F$-graph in the form (\ref{thm:cfa-nf}).
	We can then remove the loops using the identity from Def \ref{def:scfa}, to obtain a tree, namely the spider $S^n_m$. \qed
\end{proof}

Special commutative Frobenius algebras play a role in characterising orthonormal bases and classical data in quantum systems \cite{CPav,CoeckeDuncan}, and as we shall soon see, are deeply connected to GHZ states.

\begin{definition}\label{def:acfa}\em 
	An \emph{anti-special commutative Frobenius algebra} (ACFA) is a commutative Frobenius algebra where
	\begin{center}
\beginpgfgraphicnamed{antispecial_cfa}
\InputIfFileExists{antispecial_cfa.tikz}{}{\input{./figures/antispecial_cfa.tikz}}
\endpgfgraphicnamed
	\end{center}
\end{definition}

While the notion of special Frobenius algebras is standard, that of anti-special CFAs seems new. We will now characterise the normal form of an ACFA. First we show that for an ACFA $\delta$ (resp.~$\mu$) copies $\lolli$ (resp.~$\cololli$\,).

\begin{proposition}\label{prop:acfa-copy}
For any ACFA we have:

\begin{center}
\beginpgfgraphicnamed{acfa-copy-statement} 
	\circl
	\begin{tikzpicture}[dotpic,scale=0.85]
		\node [dot] (a) at (0,1) {};
		\node [dot] (b) at (0,0) {};
		\node [bn] (0) at (-1,-1) {};
		\node [bn] (1) at (1,-1) {};
		\draw (a) -- (b) -- (0) (b) -- (1);
		\draw (a) to [uploop] ();
	\end{tikzpicture}
	=
	\begin{tikzpicture}[dotpic,scale=0.85]
		\node [dot] (a) at (0,0.5) {};
		\node [dot] (b) at (1,0.5) {};
		\node [bn] (0) at (0,-0.5) {};
		\node [bn] (1) at (1,-0.5) {};
		\draw (a)--(0) (b)--(1);
		\draw (a) to [uploop] ();
		\draw (b) to [uploop] ();
	\end{tikzpicture}
\endpgfgraphicnamed 
\end{center}
\end{proposition}
\begin{proof}
	We can show this using Thm \ref{thm:cfa-nf} and anti-speciality:
	\begin{center}
\beginpgfgraphicnamed{acfa_copy_proof}
\InputIfFileExists{acfa_copy_proof.tikz}{}{\input{./figures/acfa_copy_proof.tikz}}
\endpgfgraphicnamed \quad\qed
	\end{center}
\end{proof}

Unlike with SCFAs, scalars play a role in characterising the behaviour of ACFAs. Recall from the previous section that $\circl$ is always the dimension $D$ of the underlying Hilbert space. Assuming $D > 0$, let $\icircl = 1/D$. Thus $\circl\, \icircl = 1$.

\begin{proposition}\label{pro:looploop-zero}
	For any ACFA $(H, \mult, \unit, \comult, \counit)$, either $\textrm{\rm dim}(H) = 1$ or $\dcircl = 0$.
\end{proposition}

\begin{proof}
	Let $\dcircl = k$. For $D = \textrm{dim}(H) = \circl$, we can show by Prop \ref{prop:acfa-copy} that $D k = k$.
	
	\begin{center}
\beginpgfgraphicnamed{double_loop_zero_pf}
\InputIfFileExists{double_loop_zero_pf.tikz}{}{\input{./figures/double_loop_zero_pf.tikz}}
\endpgfgraphicnamed
	\end{center}
Therefore, either $k$ is zero or $D$ is 1.
\qed
\end{proof}

From hence forth, we shall assume that $H$ is of dimension $\geq 2$, so $\dcircl$ must be zero.

\begin{theorem}\label{thm:acfa-spider}
	For an ACFA $\mathcal A$, any connected $\mathcal A$-graph is equal to one of the following:
\begin{center}
\beginpgfgraphicnamed{acfa_nf}
\InputIfFileExists{acfa_nf.tikz}{}{\input{./figures/acfa_nf.tikz}}
\endpgfgraphicnamed
\end{center}
\end{theorem}
\begin{proof}
Suppose the $\mathcal A$-graph contains more than one loop. It must then be zero.

\begin{center}
\beginpgfgraphicnamed{contains_two_loops_zero}
\InputIfFileExists{contains_two_loops_zero.tikz}{}{\input{./figures/contains_two_loops_zero.tikz}}
\endpgfgraphicnamed
\end{center}

If an $A$ graph has zero loops, it is automatically of the form of (iii.), so we shall consider only those graphs containing exactly one loop.

If the graph has zero inputs, zero outputs and one loop, it must be equal to $\circl$. Suppose it has zero inputs, at least one output, and exactly one loop. Then it must be of this form:
	\begin{center}
		\begin{tikzpicture}[dotpic]
			\node [dot] (a) at (0,1) {};
			\node [dot] (b) at (0,0) {};
			\node [bn] (0) at (-1,-1) {};
			\node at (0,-0.6) {\small ...};
			\node [bn] (1) at (1,-1) {};
			\draw (a) -- (b) -- (0) (b) -- (1);
			\draw (a) to [uploop] ();
		\end{tikzpicture}
	\end{center}
By Prop \ref{prop:acfa-copy} this can be written as:
	
	\begin{center}
		$\icircl$\,\,...\,\,$\icircl$\ $\lolli$\ ... $\lolli$
	\end{center}

The case of at least one input, zero outputs, and one loop is treated similarly. It remains only to consider the case of one or more inputs and outputs, and one loop:

\begin{center}
\beginpgfgraphicnamed{acfa-nf-proof-one-loop} 
	\begin{tikzpicture}[dotpic,yshift=-1.5cm]
		\node at (0,3.5) {\small ...};
		\node [dot] (0) at (0,3) {};
		\node [dot] (1) at (0,2) {};
		\node [dot] (2) at (0,1) {};
		\node [dot] (5) at (0,0) {};
		\node at (0,-0.5) {\small ...};

		\draw (-0.7,3.7)--(0)--(0.7,3.7) (-0.7,-0.7)--(5)--(0.7,-0.7)
			(0)--(1) (2)--(5);
		\draw [bend left] (1) to (2) (2) to (1);
	\end{tikzpicture} =
	\ \icircl
	\begin{tikzpicture}[dotpic,yshift=-1.4cm]
		\node at (0,3.5) {\small ...};
		\node [dot] (0) at (0,3) {};
		\node [dot] (1) at (0,2) {};
		\draw (1) to [downloop] ();
		\draw (-0.7,3.7)--(0)--(0.7,3.7) (0)--(1);
		
		\begin{scope}[yshift=-2mm]
			\node [dot] (2) at (0,1) {};
			\node [dot] (5) at (0,0) {};
			\node at (0,-0.5) {\small ...};
			\draw (2) to [uploop] ();
			\draw (2)--(5) (-0.7,-0.7)--(5)--(0.7,-0.7);
		\end{scope}
	\end{tikzpicture} =
	\ \icircl
		\begin{tikzpicture}[dotpic,yshift=-2.5cm]
			\node at (0,3.5) {\small ...};
			\node [dot] (0) at (0,3) {};
			\node [dot] (1) at (0,2) {};
			\draw (1) to [downloop] ();
			\draw (-0.7,3.7)--(0)--(0.7,3.7) (0)--(1);

			\begin{scope}[yshift=2cm,xshift=1.5cm]
				\node [dot] (2) at (0,1) {};
				\node [dot] (5) at (0,0) {};
				\node at (0,-0.5) {\small ...};
				\draw (2) to [uploop] ();
				\draw (2)--(5) (-0.7,-0.7)--(5)--(0.7,-0.7);
			\end{scope}
		\end{tikzpicture}
\endpgfgraphicnamed 
\end{center}

	This is then a tensor product of the two previous cases, so it can be written in the form of (iv.).
\qed
\end{proof}

\section{GHZ and W states as commutative Frobenius algebras}\label{sec:classification_states}

We have already seen that GHZ states and W states are both Frobenius states, so they each induce unique commutative Frobenius algebras. Furthermore, for CFAs on $\Q$, the conditions of  specialness and anti-specialness are actually enough to \emph{identify} GHZ and W states up the SLOCC-equivalence. The results to follow will be greatly assisted by the following lemma.

\begin{lemma}[Mathonet et al.~\cite{Mathonet2009}]\label{thm:slocc-uniform}
	If $\ket\Psi$ and $\ket\Phi$ are symmetric $N$-qubit states such that $\ket\Psi$ and $\ket\Phi$ are SLOCC-equivalent, then there exists an invertible linear map $L:\mathbb{C}^2\to \mathbb{C}^2$ such that $\ket\Psi = L^{\otimes N} \ket\Phi$.
\end{lemma}

\begin{theorem}[GHZ states are SCFAs]\label{thm:scfa-ghz}
For any special commutative Frobenius algebra on $\Q$, the induced Frobenius state is SLOCC-equivalent to $\ketGHZ$.  Furthermore, for any tripartite state $\ket\Psi$ that is SLOCC-equivalent to $\ketGHZ$, there exists $\bra\Phi$, $\bra\xi$ such that $(\ket\Psi, \bra\phi, \bra\xi)$ is a special commutative Frobenius algebra.
\end{theorem}

\begin{proof}
($\Rightarrow$) Let $\mathcal S = (\Q, \mult, \unit, \comult, \counit)$ be an SCFA. Consider the Frobenius state associated with the GHZ state.
\begin{center}
    $\ketGHZ = \ket{000}+\ket{111}$
    \qquad
    $\bra{\Phi} = \bra{00} + \bra{11} $
    \qquad
    $\bra{\xi} = \sqrt{2} \bra{+}$
\end{center}

This defines the following CFA, which we have already seen to be an SCFA.	
		\begin{center}
			$\whitecomult :: \ket{0} \mapsto \ket{00}, \ket{1} \mapsto \ket{11}$;\qquad
			$\whitecounit := \sqrt{2} \bra{+}$;\qquad
			$\whitemult := \left( \whitecomult \right)^\dagger$;\qquad
			$\whiteunit := \left( \whitecounit \right)^\dagger$
		\end{center}

We know from \cite{CPV} that an SCFA of dimension 2 is unique determined by the two linearly independent vectors copied by its comultiplication. Let $\ket u, \ket v$ be those vectors for $\comult$. Define an invertible map $L :: \ket 0 \mapsto \ket u, \ket 1 \mapsto \ket v$. We define $\comult$ in terms of $L$ and $\whitecomult$:
	
	\[
	\comult =
	\begin{tikzpicture}[dotpic]
		\node [dot,fill=white] (0) at (0,0) {};
		\node [draw=black] (L1) at (0,1) {\footnotesize $L^{-1}$};
		\node [draw=black] (L2) at (-1,-1) {\footnotesize $L$};
		\node [draw=black] (L3) at (1,-1) {\footnotesize $L$};
		\draw (0,2)--(L1)--(0) (L2)--(-1,-2) (L3)--(1,-2);
		\draw [bend right] (0) to (L2);
		\draw [bend left] (0) to (L3);
	\end{tikzpicture}\ \  ::
	\begin{cases}
	    \ket{u} \mapsto \ket{uu} \\
	    \ket{v} \mapsto \ket{vv}
	\end{cases}
	\]
	This induces $\unit$ as follows.
	
	\[\unit =
	\begin{tikzpicture}[dotpic]
		\node [dot] (0) at (0,0) {};
		\node [bn] (b2) at (-1,-1) {};

		\draw (0) to (b2);
		\draw (0) to [out=90,in=-90]
		    (0,1) to [out=90,in=90] 
		    (2,1) to [out=-90,in=90] 
	     (2,-1) to [out=-90,in=-90] 
			 (1,-1) to [out=90,in=-45] (0);
	\end{tikzpicture} =
	\begin{tikzpicture}[dotpic]
		\node [dot,fill=white] (0) at (0,0) {};
		\node [draw=black] (L1) at (0,1) {\footnotesize $L^{-1}$};
		\node [draw=black] (L2) at (-1,-1) {\footnotesize $L$};
		\node [draw=black] (L3) at (1,-1) {\footnotesize $L$};
		\draw (L1)--(0) (L2)--(-1,-2);
		\draw [bend right] (0) to (L2);
		\draw [bend left] (0) to (L3);
		\draw (L1) to [out=90,in=90] (2,1.6) to [out=-90,in=90] (2,-1.6) to [out=-90,in=-90] (L3);
	\end{tikzpicture} =
	\begin{tikzpicture}[dotpic]
		\node [dot,fill=white] (0) at (0,0) {};
		\node [draw=black] (L1) at (-1,-1) {\footnotesize $L$};
		
		\draw [bend right] (0) to (L1);
		\draw (L1) -- (-1,-2);
		\draw (0) to [out=90,in=-90]
		    (0,1) to [out=90,in=90] 
		    (2,1) to [out=-90,in=90] 
	     (2,-1) to [out=-90,in=-90] 
			 (1,-1) to [out=90,in=-45] (0);
	\end{tikzpicture} =
	\begin{tikzpicture}[dotpic]
		\node [dot,fill=white] (0) at (0,1) {};
		\node [draw=black] (L1) at (0,0) {\footnotesize $L$};
		\draw (0)--(L1)--(0,-1);
	\end{tikzpicture}
	\]
	It then follows that
	
	\[
	  \begin{tikzpicture}[dotpic,yscale=-1,yshift=5mm]
			\node [dot] (0) at (0,-1) {};
			\node [dot] (1) at (1,0) {};
			\node [bn] (b1) at (-2,1) {};
			\node [bn] (b2) at (0,1) {};
			\node [bn] (b3) at (2,1) {};
			\node [dot] (b4) at (0,-2) {};
			\draw (b1)--(0)--(b4) (b2)--(1)--(b3) (0)--(1);
		\end{tikzpicture} =
		\begin{tikzpicture}[dotpic,yscale=-1,yshift=5mm]
			\node [dot,fill=white] (0) at (0,-1) {};
			\node [dot,fill=white] (1) at (1,0) {};
			\node [draw=black] (b1) at (-2,1) {\footnotesize $L$};
			\node [draw=black] (b2) at (0,1) {\footnotesize $L$};
			\node [draw=black] (b3) at (2,1) {\footnotesize $L$};
			\node [dot,fill=white] (b4) at (0,-2) {};
			\draw (-2,2)--(b1)--(0)--(b4) (0,2)--(b2)--(1)--(b3)--(2,2) (0)--(1);
		\end{tikzpicture} = (L \otimes L \otimes L) \ketGHZ
	\]
	
\par\bigskip\noindent	($\Leftarrow$) In the other direction, we start with a symmetric state $\ket\Psi$ that is SLOCC equivalent to $\ketGHZ$. By Lemma \ref{thm:slocc-uniform}, we know $\ket\Psi$ must be of the following form, for $L$ invertible.

\begin{equation}\label{eqn:slocc-ghz-form}
	\ket\Psi := (L \otimes L \otimes L)\ketGHZ
\end{equation}
Define the following two effects.
\begin{center}
   $\bra\xi := \sqrt{2} \bra{+} L^{-1}$
   \qquad
   $\bra\Phi = (\bra{00}+\bra{11})(L^{-1} \otimes L^{-1})$
\end{center}
Then the triple $(\ket\Psi, \bra\Phi, \bra\xi)$ is a Frobenius state. Using Thm \ref{thm:frob-state-cfa}, this corresponds to the following commutative Frobenius algebra.
	\begin{eqnarray*}
		\whitecomult & = & (L \otimes L)(\ket{00}\bra{0} + \ket{11}\bra{1})L^{-1} \\
		\whitecounit\, & = & (\sqrt{2} \bra{+}) L^{-1} \\
		\whitemult & = & L(\ket{0}\bra{00} + \ket{1}\bra{11})(L^{-1} \otimes L^{-1}) \\
		\whiteunit\, & = & L (\sqrt{2}\; \ket{+})
	\end{eqnarray*}
This set of generators does indeed obey the axioms of an SCFA.
\qed
\end{proof}

\begin{theorem}[W states are ACFAs]\label{thm:acfa-w}
For any anti-special commutative Frobenius algebra on $\Q$, the induced Frobenius state is SLOCC-equivalent to $\ketW$.  Furthermore, for any tripartite state $\ket\Psi$ that is SLOCC-equivalent to $\ketW$, there exists $\bra\Phi$, $\bra\xi$ such that $(\ket\Psi, \bra\phi, \bra\xi)$ is a anti-special commutative Frobenius algebra.
\end{theorem}

\begin{proof}

\par\bigskip\noindent ($\Rightarrow$) Let	$\mathcal A = (\Q, \mult, \unit, \comult, \counit)$ be an ACFA. Since $\comult$ is left and right unital, it is not separable i.e. it cannot be expressed as one of these forms: $\ket{A}\bra{b}$, $\ket{a} \otimes L$, $L \otimes \ket{a}$. 	Note that  $\lolli = \textrm{Tr}_{\Q}(\comult)$, the result of tracing the input to the right leg. Assume without loss of generality that $\lolli$ is normal, which can be achieved by rescaling $\comult$ with some scalar $\lambda$. To avoid confusion we denote these recalled variants of $\comult$ and $\lolli$ as $\delta$ and $\ket t$. Let $B:= \left\{ \ket t, \ket{t^\perp} \right\}$ be an orthonormal basis for $\mathbb C^2$. By Prop \ref{prop:acfa-copy} we have $\delta \ket t = a \ket{tt}$. So, for some $\ket\Psi$:
	\[ \delta = a \ket{tt}\bra{t} + \ket{\Psi}\bra{t^\perp} \] 
Take the right partial trace of both sides:
	\[ \ket t = a \ket t + \textrm{Tr}_{\Q}(\ket\Psi \bra{t^\perp}) \]
So, $\textrm{Tr}_{\Q}(\ket\Psi \bra{t^\perp}) = (1-a)\ket t$. We can express $\ket\Psi$ in the basis $B$:
	\[ \ket\Psi = \ket{ut} + \ket{vt^\perp} \]
and
	\[ (1-a)\ket t = \textrm{Tr}_{\Q}(\ket\Psi \bra{t^\perp}) = \ket{v} \]
Now, letting $\ket u = b\ket t + c\ket{t^\perp}$,
	\[ \ket\Psi = b \ket{tt} + c \ket{t^\perp t} + (1-a)\ket{tt^\perp} \]
Plugging in to $\delta$, letting $d = (1-a)$:
	\[ \delta = a \ket{tt}\bra{t} + \left( b\ket{tt} + c \ket{t^\perp t} + d \ket{t t^\perp} \right) \bra{t^\perp} \]
$d \neq 0$, otherwise $\delta$ is separable. Let $\ket s = \frac{1}{d}(b \ket{t} + c\ket{t^\perp})$: 
	\[ \delta = a \ket{tt}\bra{t} + \left( d \ket{st} + d \ket{t t^\perp} \right) \bra{t^\perp} \]
Note that $\ket s \neq k \ket t$, otherwise $\delta$ is separable. Also, by definition of Frobenius algebra, $\delta \eta$ is non-degenerate (i.e. entangled). So, choose non-proportional $\ket{s'}, \ket{t'}$ such that:
	\[ (\bra t \otimes 1)\delta \eta = \frac{1}{a} \ket{s'}\ \ \ \ \ \ \ (\bra{t^\perp} \otimes 1)\delta \eta = \frac{1}{d} \ket{t'}\]
Then, the state associated with $(\delta, \eta)$ is:
	\[
	  \begin{tikzpicture}[dotpic,xscale=-1,yscale=-1,yshift=5mm]
			\node [dot] (0) at (0,-1) {};
			\node [dot] (1) at (1,0) {};
			\node [bn] (b1) at (-2,1) {};
			\node [bn] (b2) at (0,1) {};
			\node [bn] (b3) at (2,1) {};
			\node [dot] (b4) at (0,-2) {};
			\draw (b1)--(0)--(b4) (b2)--(1)--(b3) (0)--(1);
		\end{tikzpicture} =\lambda 
	(\delta \otimes 1) \delta \eta = \lambda(\ket{tts'} + \ket{stt'} + \ket{tt^\perp t'}) \]
Now, define the following local maps:
	\[ L_1 ::  \ket t \mapsto \ket 0, \ket s \mapsto \ket 1  \]
	\[ L_2 ::  \ket t \mapsto \ket 0, \ket{t^\perp} \mapsto \ket 1  \]
	\[ L_3 ::  \ket{t'} \mapsto \ket 0, \ket{s'} \mapsto \ket 1 \]
These are all invertible because they take bases to bases. Then
	\[ (L_1 \otimes L_2 \otimes L_3)(\delta \otimes 1) \delta \eta = \ketW \]

\par\bigskip\noindent ($\Leftarrow$) For the converse, we mirror the construction from the proof of the SCFA-GHZ case. The only difference is the choice of cup and counit. Start with a symmetric state $\ket\Psi$ that is SLOCC equivalent to $\ketW$. By Lemma \ref{thm:slocc-uniform}, we know $\ket\Psi$ must be of the following form, for $L$ invertible.
\begin{equation}\label{eqn:slocc-w-form}
	\ket\Psi := (L \otimes L \otimes L)\ketW
\end{equation}

Define the following two effects.
\begin{center}
   $\bra\xi := \bra{0} L^{-1}$
   \qquad
   $\bra\Phi = (\bra{01}+\bra{10})(L^{-1} \otimes L^{-1})$
\end{center}
Then the triple $(\ket\Psi, \bra\Phi, \bra\xi)$ is a Frobenius state. Using Thm \ref{thm:frob-state-cfa}, we construct the CFA.
\begin{eqnarray*}
	\comult & := & (L \otimes L)(\ket{10}\bra{1} + \ket{01}\bra{1} + \ket{00}\bra{0})L^-1 \\
	\counit\, & := & \bra{0} L^{-1} \\
	\mult & := & L(\ket{1}\bra{11} + \ket{0}\bra{01} + \ket{0}\bra{10})(L^{-1} \otimes L^{-1}) \\
	\unit\, & := & L \ket{1}
\end{eqnarray*}
It is then easy to check that this is an ACFA.
\qed
\end{proof}

\begin{corollary}
Any SLOCC-maximal tripartite qubit state is SLOCC-equivalent to a Frobenius state.
\end{corollary}

\section{Classification of commutative Frobenius algebras on $\Q$}\label{sec:classification_CFA}

We know from \cite{DVC} that $\ketGHZ$ and $\ketW$ are the \emph{only} genuine tripartite qubit states, up to SLOCC-equivalence. Thus, the result above offers an exhaustive classification of CFA's on $\Q$, up to local maps.

\begin{corollary}\label{cor:cfa-local-equiv}
For any CFA on $\Q$ there exists an invertible linear map $L:\mathbb{C}^2\to \mathbb{C}^2$ such that the following maps define a CFA that is either special or anti-special.
\begin{equation}\label{eq:local-scfa-acfa}
\beginpgfgraphicnamed{local_equiv}
\InputIfFileExists{local_equiv.tikz}{}{\input{./figures/local_equiv.tikz}}
\endpgfgraphicnamed
\end{equation}
\end{corollary}

\begin{proof}
    Every CFA on $\mathbb{C}^2$ induces a Frobenius state $S_3^0 = \threestate$. Frobenius states are SLOCC-maximal, so it must be SLOCC-equivalent to either $\ketGHZ$ or $\ketW$, so we can apply the construction in the proofs of Thms \ref{thm:scfa-ghz} and \ref{thm:acfa-w}. This preserves $\threestate$, yet yields a new cup:
	\begin{tikzpicture}[dotpic,scale=0.5,yscale=-1]
		\draw (0,-0.5) to [out=90,in=180] (1,0.5) to [out=0,in=90] (2,-0.5);
		\node [white dot,inner sep=0pt] (0) at (1,0.5) {$\cdot$};
	\end{tikzpicture}. Let $L$ then be the composition of the old cap with the new cup.
	\[
	\begin{tikzpicture}[dotpic]
		\node [square box,fill=none,inner sep=0pt] (0) at (0,0) {\small L};
		\draw (0,1)--(0)--(0,-1);
	\end{tikzpicture} =
	\begin{tikzpicture}[dotpic,yshift=-5mm,scale=0.5]
			\node [bn] (b0) at (-1,2) {};
			\node [white dot,inner sep=0pt] (0) at (0,0) {$\cdot$};
			\node [dot] (1) at (1.5,1) {};
			\node [bn] (b1) at (2.5,-1) {};
			\draw (b0) to [out=-90,in=180] (0) (0) to [out=0,in=180] (1) (1) to [out=0,in=90] (b1);
		\end{tikzpicture}
	\]
    By composing the multiplication and comultiplication with $L$ as in Eq \ref{eq:local-scfa-acfa}, the resulting CFA will be either special or anti-special.
\qed
\end{proof}

In two dimensions, any unital algebra is automatically commutative, so another simple consequence of this classification of tripartite states is that \emph{every} unital algebra on $\Q$ admits a counit $\epsilon$ that extends the monoid to a CFA.

\begin{corollary}\label{cor:2d-frobenius-form}
	Every unital algebra on $\mathbb C^2$ extends to a CFA.
\end{corollary}

\begin{proof}
    Since all two-element monoids are commutative, every unital algebra on $\mathbb C^2$ is commutative, and the unit law forces the multiplication map to be maximal with respect to local operations. It must therefore be locally equivalent to $\ketGHZ$ or $\ketW$. In either case, choosing the appropriate counit will make the unital algebra into a commutative Frobenius algebra.
\qed
\end{proof}

All of these results now follow easily from the fact that $\ketGHZ$ and $\ketW$ are the only SLOCC-maximal tripartite states for qubits. These are examples of how results from entanglement theory can be translated straight into results about algebras using the notion of a Frobenius state.

\section{$N$-partite entanglement from interacting GHZ and W states}\label{sec:interacting-states}

Thms \ref{thm:scfa-spider} and \ref{thm:acfa-spider} show that using a single ACFA or SCFA, we can construct relatively few states, namely those states that are SLOCC-equivalent to products of $\ket{\textit{GHZ}_n}$ and $\ket{\textit{W}_n}$. However, when we compose these two structures, a wealth of states emerge.

We begin by looking at the SCFA $\mathcal G$ from Ex \ref{ex:concreteGHZ} and the ACFA $\mathcal W$ from Ex \ref{ex:concreteW}. These satisfy many concrete identities when we compose them, but we shall soon see that the following four identities suffice to identify these two commutative Frobenius algebras, up to a change of basis.

\begin{definition}\label{def:ghz-w-pair}\em
    A special commutative Frobenius algebra $\mathcal S = (H, \whitemult, \whiteunit, \whitecomult, \whitecounit)$ and an antispecial commutative Frobenius algebra $\mathcal A = (H, \mult, \unit, \comult, \counit)$ are said to define a \emph{GHZ/W-pair} if the following equations hold:
\begin{center}
	{\rm (i.)}\ \ %
\beginpgfgraphicnamed{ghz_w_ax1}
\InputIfFileExists{ghz_w_ax1.tikz}{}{\input{./figures/ghz_w_ax1.tikz}}
\endpgfgraphicnamed
	\qquad\quad
	{\rm (ii.)} %
\beginpgfgraphicnamed{ghz_w_ax2}
\InputIfFileExists{ghz_w_ax2.tikz}{}{\input{./figures/ghz_w_ax2.tikz}}
\endpgfgraphicnamed

	{\rm (iii.)} %
\beginpgfgraphicnamed{ghz_w_ax3}
\InputIfFileExists{ghz_w_ax3.tikz}{}{\input{./figures/ghz_w_ax3.tikz}}
\endpgfgraphicnamed
	\qquad\qquad\qquad
	{\rm (iv.)}\ \ %
\beginpgfgraphicnamed{ghz_w_ax4}
\ensuremath{\circl \tickunit =\, \lolli}}
\endpgfgraphicnamed \ \ \,
\end{center}
\end{definition}

We shall now show that when $H = \Q$, these equations suffice to uniquely identify the SCFA ${\cal G}$ of Example \ref{ex:concreteGHZ} and the ACFA ${\cal W}$ of Example \ref{ex:concreteW}, up to a change of basis. Before we show this, we first develop some basic facts about a GHZ/W-pair.

\begin{lemma}\label{lem:dot-lolli-2d}
    If $\dim(H) \geq 2$, the points $\unit$ and $\lolli$ span a 2-dimensional space.
\end{lemma}

\begin{proof}
    If $\unit$ were proportional to $\lolli$ (written $\unit \approx \lolli$), anti-specialness would force the identity map to be rank 1.
    \begin{center}
\beginpgfgraphicnamed{dot_lolli_2d}
\InputIfFileExists{dot_lolli_2d.tikz}{}{\input{./figures/dot_lolli_2d.tikz}}
\endpgfgraphicnamed
    \end{center}
    For $\dim(H) \geq 2$, this is a contradiction.
\end{proof}

Conditions (i.)--(iv.) imply that $\unit$ and $\lolli$ are both copiable points (up to a scalar) of the SCFA and that a ``tick'' is a self-inverse permutation of them. We show that $\lolli$ is copied by $\whitecomult$ as follows.
\begin{equation}\label{eq:loopcopy}
\beginpgfgraphicnamed{loop_copy_proof} 
	\circl\hspace{-3mm}
	\begin{tikzpicture}[dotpic]
		\node [dot] (g1) at (0,1) {};
		\node [white dot] (r1) at (0,0) {}; 
		\node [bn] (c1) at (-1,-1) {};
		\node [bn] (c2) at (1,-1) {};
		\draw [uploop] (g1) to ();
		\draw (g1)--(r1)--(c1) (r1)--(c2);
	\end{tikzpicture} =\ 
	\circl\,\circl\hspace{-3mm}
	\begin{tikzpicture}[dotpic]
		\node [dot] (g1) at (0,1) {};
		\node [white dot] (r1) at (0,0) {};
		\node [bn] (c1) at (-1,-1) {};
		\node [bn] (c2) at (1,-1) {};
		\draw (g1)-- node[tick]{-} (r1)--(c1) (r1)--(c2);
	\end{tikzpicture} =\ 
	\circl\,\circl\hspace{-3mm}
	\begin{tikzpicture}[dotpic]
		\node [dot] (g1) at (0,1) {};
		\node [white dot] (r1) at (0,0) {};
		\node [bn] (c1) at (-1,-1) {};
		\node [bn] (c2) at (1,-1) {};
		\draw (g1)--(r1)-- node[tick]{-} (c1) (r1)-- node[tick]{-} (c2);
	\end{tikzpicture} =\ 
	\circl \tickunit \circl \tickunit =\ 
	\lolli\ \lolli
\endpgfgraphicnamed 
\end{equation}

We note also that the ``tick'' leaves the counit of $\mathcal S$ invariant.

\begin{lemma}\label{lem:white-unit-tick}
	$\whitecounit = \whitetickcounit$.
\end{lemma}

\begin{proof}
	From condition (i.) we conclude that $\tick$ is self-inverse. The result then follows from condition (ii.).
	\begin{center}
	\beginpgfgraphicnamed{white_tick_inv_proof} 
		\begin{tikzpicture}[dotpic]
			\node [white dot] (0) at (0,-0.4) {};
			\draw (0) -- (0,0.4);
		\end{tikzpicture} =
		\begin{tikzpicture}[dotpic]
			\node [white dot] (0) at (0,-0.4) {};
			\draw (0) -- node[tick,pos=0.3]{-} node[tick,pos=0.7]{-} (0,0.4);
		\end{tikzpicture} =
		\begin{tikzpicture}[dotpic]
			\node [white dot] (1) at (-0.4,-0.4) {};
			\node [white dot] (2) at (0.4,-0.4) {};
			\node [white dot] (0) at (0,0) {};
			\draw (1) -- node[tick]{-} (0);
			\draw (2) -- (0);
			\draw (0) -- node[tick]{-} (0,0.5);
		\end{tikzpicture} =
		\begin{tikzpicture}[dotpic]
			\node [white dot] (1) at (-0.4,-0.4) {};
			\node [white dot] (2) at (0.4,-0.4) {};
			\node [white dot] (0) at (0,0) {};
			\draw (1) -- node[tick,pos=0.3]{-} node[tick,pos=0.7]{-} (0);
			\draw (2) -- node[tick]{-} (0);
			\draw (0) -- (0,0.5);
		\end{tikzpicture} =
		\begin{tikzpicture}[dotpic]
			\node [white dot] (1) at (-0.4,-0.4) {};
			\node [white dot] (2) at (0.4,-0.4) {};
			\node [white dot] (0) at (0,0) {};
			\draw (1) -- (0);
			\draw (2) -- node[tick]{-} (0);
			\draw (0) -- (0,0.5);
		\end{tikzpicture} =
		\begin{tikzpicture}[dotpic]
			\node [white dot] (0) at (0,-0.4) {};
			\draw (0) -- node[tick]{-} (0,0.4);
		\end{tikzpicture}
	\endpgfgraphicnamed \quad \qed
	\end{center}

\end{proof}

	Given a commutative Frobenius algebra $\blackdot = (H, \mult, \unit, \comult, \counit)$, we can define an operation \em $\blackdot$-transpose \em of a morphism
	\[
	f: H \otimes \ldots \otimes H \rightarrow H \otimes \ldots \otimes H
	\]
	as follows:
	\begin{equation*}
\beginpgfgraphicnamed{black_transpose}
\InputIfFileExists{black_transpose.tikz}{}{\input{./figures/black_transpose.tikz}}
\endpgfgraphicnamed
	\end{equation*}
Using Frobenius identities, it is easy to verify that:
    \begin{itemize}
	    \item[--] $(1_A)^{\,\blackdot T} = 1_A$,
	    \item[--] $(f \circ g)^{\,\blackdot T}
	            = g^{\,\blackdot T} \circ f^{\,\blackdot T}$,
	    \item[--] $(f\otimes g)^{\,\blackdot T}
	            = f^{\,\blackdot T} \otimes\ g^{\,\blackdot T}$, and
	    \item[--] $(f^{\,\blackdot T})^{\,\blackdot T} = f$.
	\end{itemize}
	
	These show that $\blackdot$-transpose is a particularly well-behaved functional that respects the tensor product and reverses the composition of linear maps. We justify the name by noting that when the copiable points of $\whitecomult$ form an orthonormal basis $B$, this is a normal transpose (in $B$) and a change of basis given by conjugating with ticks.
	
	In category theoretic terms, the conditions above mean $\blackdot$-transpose extends to a monoidal, involutive, contravariant endofunctor on the subcategory of finite-dimensional Hilbert spaces given by tensor copies of $A$.

\begin{lemma}\label{lem:upside-down}
	$(\tick)^{\blackdot T}\! =\! \tick$, $(\mult)^{\blackdot T}\! = \comult$, $(\unit)^{\blackdot T} = \counit$, $(\whiteunit)^{\blackdot T} = \whitecounit$, $(\whitemult)^{\blackdot T} = \whitecomult$.
\end{lemma}

\begin{proof}
	$(\tick)^{\blackdot T} = \tick$ holds by condition (i.). $(\mult)^{\blackdot T} = \comult$ and $(\unit)^{\blackdot T} = \counit$ are true by Frobenius identities. $(\whitecounit)^{\blackdot T} = \whiteunit$ by Lem \ref{lem:white-unit-tick}. The final identity holds by condition (ii.).
	\begin{center}
	\beginpgfgraphicnamed{white_transpose_inv} 
		\(\left(\begin{tikzpicture}[dotpic]
			\node [white dot] (a) {};
			\draw (a) -- (0, -0.5);
			\draw (a) -- (-0.4,0.4);
			\draw (a) -- (0.4,0.4);
		\end{tikzpicture}\right)^{\blackdot T}\) =
		\begin{tikzpicture}[dotpic]
			\node [white dot] (a) {};
			\draw (a) -- node[tick]{-} (0, 0.5);
			\draw (a) -- node[tick]{-} (-0.4,-0.4);
			\draw (a) -- node[tick]{-} (0.4,-0.4);
		\end{tikzpicture} = 
		\begin{tikzpicture}[dotpic]
			\node [white dot] (a) {};
			\draw (a) -- (0, 0.5);
			\draw (a) -- node[tick,pos=0.3]{-} node[tick,pos=0.7]{-} (-0.4,-0.4);
			\draw (a) -- node[tick,pos=0.3]{-} node[tick,pos=0.7]{-} (0.4,-0.4);
		\end{tikzpicture} =
		\begin{tikzpicture}[dotpic]
			\node [white dot] (a) {};
			\draw (a) -- (0, 0.5);
			\draw (a) -- (-0.4,-0.4);
			\draw (a) -- (0.4,-0.4);
		\end{tikzpicture}
	\endpgfgraphicnamed \quad \qed 
	\end{center}
\end{proof}

Since $f = g \iff f^{\blackdot T} = g^{\blackdot T}$, this lemma effectively gives us a way to turn any known identity upside-down.

Next, we have a lemma on scalars.

\begin{lemma}\label{lem:scalars}
	\(
	\begin{tikzpicture}[dotpic]
		\node [dot] (0) at (0,0.25) {};
		\node [dot] (1) at (0,-0.25) {};
		\draw (0) -- node[tick]{-} (1);
	\end{tikzpicture} =
	\begin{tikzpicture}[dotpic]
		\node [dot] (0) at (0,0.25) {};
		\node [white dot] (1) at (0,-0.25) {};
		\draw (0) -- node[tick]{-} (1);
	\end{tikzpicture} =
	\begin{tikzpicture}[dotpic]
		\node [dot] (0) at (0,0.25) {};
		\node [white dot] (1) at (0,-0.25) {};
		\draw (0) -- (1);
	\end{tikzpicture}
	= 1
	\)
\end{lemma}

\begin{proof}
	These equations follow from the fact that $\circl$ ($= \dim(H)$) admits an inverse $\icircl$.

	\begin{center}
\beginpgfgraphicnamed{unbiased_identity_scalar}
\InputIfFileExists{unbiased_identity_scalar.tikz}{}{\input{./figures/unbiased_identity_scalar.tikz}}
\endpgfgraphicnamed \quad \qed
	\end{center}
\end{proof}

$\mathcal G$ is the SCFA derived from $\ketGHZ$ and $\mathcal W$ is the ACFA derived from $\mathcal W$. By construction, $(\mathcal G, \mathcal W)$ forms a GHZ/W pair. We can now prove that this is the only GHZ/W pair on $\Q$, up to a change of basis.

\begin{theorem}\label{thm:HilbSCFA-ACFA-correspondence}
    For any GHZ/W-pair $(\mathcal S, \mathcal A)$ on $\Q$, there exists a change of basis $L$ that turns the canonical GHZ/W-pair $(\mathcal G, \mathcal W)$ into $(\mathcal S, \mathcal A)$. Furthermore, fixing an SCFA or ACFA uniquely determines a GHZ/W-pair, up to a possible permutation of basis vectors.
\end{theorem}


\begin{proof}
First, we fix a SCFA $\mathcal S = (\Q, \whitemult, \whiteunit, \whitecomult, \whitecounit)$. Let $\ket{e_i}$ be the set of copiable points of $\whitecomult$, i.e. points such that $\whitecomult \circ \ket{e_i} = \ket{e_i} \otimes \ket{e_i}$. We know from \cite{CPV}, that for a special Frobenius algebra, such points span the entire space, which in this case is $\Q$. We now show that, up to permutation of the points $\ket{e_0} \leftrightarrow \ket{e_1}$, conditions (i.)--(iv.) uniquely determine the ACFA $(\mult, \unit, \comult, \counit)$.

By conditions (ii.) and (iii.) and Lem \ref{lem:dot-lolli-2d}, we know that $\unit$ and $\tickunit$ are distinct copiable points, so let $\ket{e_0} = \unit$ and $\tickunit = \ket{e_1}$. By condition (i.) the tick is an involution, so it must be the permutation $\ket{e_0} \leftrightarrow \ket{e_1}$. Therefore we have defined the black cap,
\beginpgfgraphicnamed{black_cap_identity} 
	\begin{tikzpicture}[dotpic,scale=0.5]
		\draw (0,-0.5) to [out=90,in=180] (1,0.5) to [out=0,in=90] (2,-0.5);
		\node [dot] (0) at (1,0.5) {};
	\end{tikzpicture} :=
	\begin{tikzpicture}[dotpic,scale=0.5]
		\draw (0,-0.5) to [out=90,in=180] (1,0.5) to [out=0,in=90] node[tick]{-} (2,-0.5);
		\node [white dot] (0) at (1,0.5) {};
	\end{tikzpicture}
\endpgfgraphicnamed. 
Furthermore, by condition (iv.) and anti-specialness, we have
\begin{center}
\beginpgfgraphicnamed{tick_unit_copy} 
	\begin{tikzpicture}[dotpic,scale=0.7]
		\node [dot] (1) at (0,1) {};
		\node [dot] (0) at (0,0) {};
		\draw (1) -- node[tick]{-} (0);
		\draw (-0.7,-0.7) -- (0) -- (0.7,-0.7);
	\end{tikzpicture} =
	\icircl
	\begin{tikzpicture}[dotpic,scale=0.7]
		\node [dot] (1) at (0,1) {};
		\node [dot] (0) at (0,0) {};
		\draw [uploop] (1) to ();
		\draw (1) -- (0);
		\draw (-0.7,-0.7) -- (0) -- (0.7,-0.7);
	\end{tikzpicture} =
	\icircl\,\icircl\ \lolli\ \lolli\ 
	= \tickunit\tickunit
\endpgfgraphicnamed 
\end{center}
Now, we have completely defined $\comult$.
\[
\comult :: \left\{\begin{array}{cl}
	 \unit & \mapsto \ \begin{tikzpicture}[dotpic,scale=0.5]
		\draw (0,-0.5) to [out=90,in=180] (1,0.5) to [out=0,in=90] (2,-0.5);
		\node [dot] (0) at (1,0.5) {};
	\end{tikzpicture} \\
	\tickunit & \mapsto \ \tickunit\tickunit
\end{array}\right.
\]
The data $(\comult, \unit)$ suffice to define $\mathcal A$. Writing these out symbolically, we have:
\begin{equation}\label{eq:acfa-form-1}
		\comult :: \ket{e_0} \mapsto \ket{e_0,e_0},\ \  
		           \ket{e_1} \mapsto \ket{e_0,e_1} + \ket{e_1,e_0}
		\qquad\qquad
		\unit = \ket{e_1}
		\qquad
\end{equation}

Applying the change of basis $L :: \ket{e_i} \mapsto \ket{i}$ turns the pair $(\mathcal S, \mathcal A)$ into $(\mathcal G, \mathcal W)$. Of course, we could have chosen $\ket{e_0}= \tickunit$. This would induce the this ACFA, which we'll call $\mathcal A'$:
\begin{equation}\label{eq:acfa-form-2}
		\comult :: \ket{e_1} \mapsto \ket{e_1,e_1},\ \  
		           \ket{e_0} \mapsto \ket{e_1,e_0} + \ket{e_0,e_1}
		\qquad\qquad
		\unit = \ket{e_0}
		\qquad
\end{equation}

Then, the change of basis $L' :: \ket{e_i} \mapsto \ket{1-i}$ turns the pair $(\mathcal S, \mathcal A')$ into $(\mathcal G, \mathcal W)$. Note also that $\mathcal A$ and $\mathcal A'$ are related by the permutation $e_i \leftrightarrow e_{1-i}$, i.e. the map $\tick$.

Conversely, fix an ACFA $\mathcal A = (\Q, \mult, \unit, \comult, \counit)$. We have already established that the points $\{ \unit, \lolli \}$ span $\Q$. By condition (iii.) and Eq \ref{eq:loopcopy} the basis $\{ \unit, \lolli \}$ completely determines $\whitecomult$.
\[
\whitecomult ::  \left\{\begin{array}{cl}
	\, \unit & \mapsto \ \unit\ \unit \\
	 \lolli & \mapsto \ \icircl\ \lolli\ \lolli
\end{array}\right.
\]
By Lemmas \ref{lem:white-unit-tick} and \ref{lem:scalars}, we can see that $\{ \counit, \cololli \}$ totally determines $\whiteunit$.
\[
\whiteunit :: \left\{\begin{array}{cl}
	 \,\counit & \mapsto \ 1_{\mathbb C} \\
	 \cololli  & \mapsto \ \circl
\end{array}\right.
\]
The data $(\whitecomult, \whiteunit)$ suffice to define $\mathcal S$. This implies that $\mathcal A$ is of the form of either Eq \ref{eq:acfa-form-1} or Eq \ref{eq:acfa-form-2} for the copiable points $e_i$ of $\mathcal S$. In either case, we can choose a suitable $L$.
\qed
\end{proof}


\subsection{Building states from a GHZ/W-pair and universality}\label{sec:generating-states}

As mentioned above, there is necessarily an infinite number of SLOCC classes when $N\geq 4$ \cite{DVC}, and to obtain finite classification results one considers \emph{super-classes}.  An example of this approach is \cite{Lamata}, where the authors introduce a classification scheme based upon the right singular subspace of a pure state. They begin with the observation that a column vector with $2^N$ entries has the same data as a $2^{(N-1)} \times 2$ matrix. Therefore, they treat a pure state on $N$ qubits as a map from $\bigotimes^{(N-1)}\Q$ to $\Q$. Performing a singular value decomposition on such a matrix yields a 1- or 2-dimensional right singular subspace, spanned by vectors in $\bigotimes^{N-1}\mathbb{C}^2$. The SLOCC super-class of this state is then labeled by the SLOCC super-classes of these spanning vectors, thus performing the inductive step. The base case is $\mathbb C^2\otimes \mathbb{C}^2$, where the only two SLOCC classes are represented by the product state and the Bell state.

An alternative way of looking at this scheme is to consider $N$-partite states as ``controlled'' $(N-1)$-partite states. That is to say, the right singular space of a state  is spanned by $\{ \ket\Psi, \ket\Phi \}$ iff there exists a SLOCC-equivalent state of the form $\ket{0 \Psi} + \ket{1 \Phi}$. This provides an operational description of a SLOCC superclass. Namely, a state $\ket\Theta$ is in the SLOCC superclass $\{ \ket\Psi, \ket\Phi \}$ if there exists \emph{some} (two-dimensional, possibly non-orthogonal) basis $B$ such that performing a (generalised) measurement in $B$ of the first qubit yields a state that is SLOCC-equivalent to $\ket\Psi$ for outcome 1 and $\ket\Phi$ for outcome 2.

From this point of view, we can show that the language of GHZ/W-pairs realises the inductive step. We do this by realising that SCFAs perform a role analogous to junctions in classical circuits, and ACFAs analogous to \emph{tristates}, or electronic switches. Suppose we identify $\unit$ with the bit 1 and $\tickunit$ with the bit 0. For an input of $\unit$ or $\tickunit$, the map $\whitecomult$ merely copies it. However, for an input of $\unit$, the map $\comult$ forms an entangled state $\ket{01} + \ket{10}$, but for $\tickunit$, it separates. This behaves a bit like a tristate:
\begin{equation}
    \includegraphics[scale=0.7]{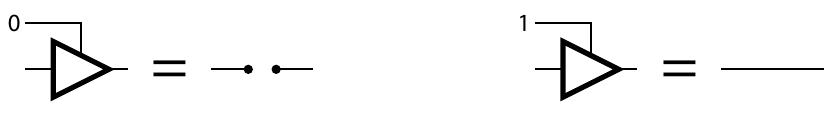}
\end{equation}

\begin{figure}
    \centering
        \includegraphics[scale=1]{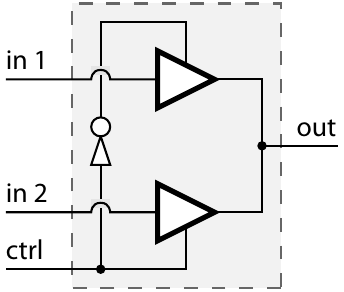}\qquad \qquad
        \includegraphics[scale=1]{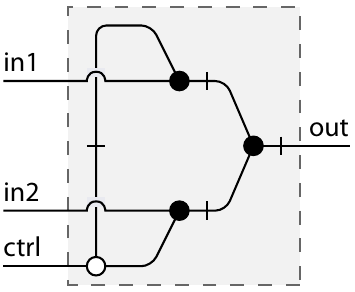}
    \caption{Classical and quantum multiplexors}
    \label{fig:both-muxes}
\end{figure}

Using this analogy, we can construct a quantum two-way switch, or multiplexor, following almost exactly the design of a classical multiplexor (see Fig \ref{fig:both-muxes}). We'll call this construction QMUX.

\begin{theorem}\label{thm:inductive step}
When considering the GHZ/W-pair on $\mathbb{C}^2$ as in Eqs (\ref{GHZ-SCFA}) and  (\ref{W-ACFA}), the linear map
\begin{equation}
	\raisebox{-17pt}{\includegraphics[scale=0.7]{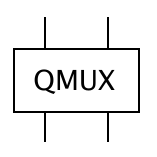}}
	:=
	\begin{tikzpicture}[dotpic,yshift=0.5cm]
		\node [bn] (d1) at (0,0.7) {};
		\node [bn] (d2) at (1,0.7) {};
		\node [dot] (0) at (0,0) {};
		\node [dot] (1) at (1,0) {};
		\node [white dot] (2) at (0,-1) {};
		\node [dot] (3) at (1,-1) {};
		\node [bn] (c1) at (0,-1.7) {};
		\node [bn] (c2) at (1,-1.7) {};
		\draw
			(d1)--(0)-- node[tick]{-} (2)-- node[tick]{-} (c1)
			(d2)--(1)-- node[tick]{-} (3)-- node[tick]{-} (c2)
			(0)-- node[tick,pos=0.7]{-} (3)  (1)--(2);
	\end{tikzpicture}\vspace{-1mm}
\end{equation}
takes states $\ket\psi \otimes \ket\phi$ to a state that is SLOCC-equivalent to $\braket{1}{\phi}\ket{0\psi} + \braket{1}{\psi}\ket{1\phi}$. From this it follows that more generally,\vspace{-2mm}
\begin{center}
	\includegraphics[scale=0.7]{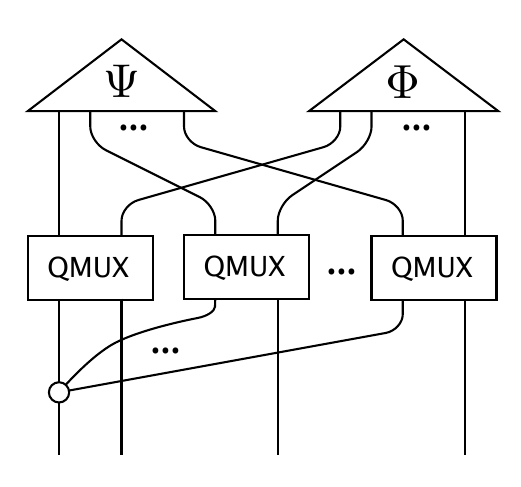}\vspace{-2mm}
\end{center}
takes the states $\ket\Psi \otimes \ket\Phi\in\left(\bigotimes^{N-1}\mathbb{C}^2\right)\otimes\left(\bigotimes^{N-1}\mathbb{C}^2\right)$ to
	\[ 
	\braket{\underbrace{1\ldots 1}_{N-1}}{\Phi}\ket{0\Psi} + 
	\braket{\underbrace{1\ldots 1}_{N-1}}{\Psi}\ket{1\Phi}\,.\vspace{-1mm}
	\]
\end{theorem}	 
\begin{proof}	
We show this by using conditions (i.)--(iv.). We only require the result to hold up to SLOCC-equivalence, so we shall disregard the scalars $\circl$ and $\icircl$. Note that $\bra{0} = \icircl\, \cololli$ and $\bra{1} = \counit$:
	\begin{center}
\beginpgfgraphicnamed{qmux_plug_dot}
\InputIfFileExists{qmux_plug_dot.tikz}{}{\input{./figures/qmux_plug_dot.tikz}}
\endpgfgraphicnamed
	\end{center}
	
	\begin{center}
\beginpgfgraphicnamed{qmux_plug_lolli}
\InputIfFileExists{qmux_plug_lolli.tikz}{}{\input{./figures/qmux_plug_lolli.tikz}}
\endpgfgraphicnamed
	\end{center}
	
	We shall explain the first sequence of equalities. The second proceeds similarly. Step 1 is due to axiom (iv.). For step 2, we can show that $\cololli$ is a copiable point of $\whitemult$ using Lem \ref{lem:upside-down}. Step 3 is due to anti-specialness and step 4 to axiom (iv.). Step 5 is two applications of the unit and axiom (i.), which implies that $\tick$ is involutive.
\qed
\end{proof}


Scalars $\braket{{1\ldots 1}}{\Psi}$ and $\braket{{1\ldots 1}}{\Phi}$ are be assumed to be non-zero. If this is not the case we vary the representatives of SLOCC-classes. It is an easy exercise to show that any state is SLOCC-equivalent to a state that is not orthogonal to $\ket{1\ldots 1}$.

\begin{theorem}\label{thm:arblinear}
When considering the GHZ/W-pair on $\mathbb{C}^2$ as in Eqs (\ref{GHZ-SCFA}) and  (\ref{W-ACFA}), an arbitrary linear map $L:\mathbb{C}^2\to\mathbb{C}^2$ can expressed as:
	\begin{center}
\beginpgfgraphicnamed{ldu_decomp}
\InputIfFileExists{ldu_decomp.tikz}{}{\input{./figures/ldu_decomp.tikz}}
\endpgfgraphicnamed
	\end{center}
for some single-qubit states $\psi$, $\phi$ and $\xi$.
\end{theorem}
\begin{proof}
For $A$ defined as above on the left, any 1-qubit linear map can be expressed as $A$ or $A \circ \left(\tick\right)$. To see why, recall that an arbitrary $2 \times 2$ matrix admits a decomposition:
\begin{equation}\label{eqn:ldu}
	A = P L D U
\end{equation}
Where $P$ is a permutation, $L$ and $U$ are unit-diagonal lower-triangular and upper-triangular matrices, and $D$ is a diagonal matrix.

For a vectors $\ket\psi, \ket\phi, \ket\xi \in \mathbb C^2$, we can construct the following maps:
\begin{center}
	\(L := 
	\begin{tikzpicture}[dotpic,scale=0.7,yshift=5mm]
		\node [pt] (0) at (0,0) {\footnotesize $\xi$};
		\node [dot] (1) at (1,-1) {};
		\draw (0) to [out=-90,in=180] node [tick] {-} (1);
		\draw (1,0.5)-- node [tick,pos=0.66] {-} (1)-- node [tick] {-} (1,-2);
	\end{tikzpicture}
	 = \left(\begin{matrix} \xi_2 & 0\\ \xi_1 & \xi_2 \end{matrix}\right)\)
	\qquad
	\(D := 
	\begin{tikzpicture}[dotpic,scale=0.7,yshift=5mm]
		\node [pt] (0) at (0,0) {\footnotesize $\phi$};
		\node [white dot] (1) at (1,-1) {};
		\draw (0) to [out=-90,in=180] (1);
		\draw (1,0.5)--(1)--(1,-2);
	\end{tikzpicture}
	= \left(\begin{matrix} \phi_1 & 0\\ 0 & \phi_2 \end{matrix}\right)\)
	\qquad
	\(U :=
	\begin{tikzpicture}[dotpic,scale=0.7,yshift=5mm]
		\node [pt] (0) at (0,0) {\footnotesize $\psi$};
		\node [dot] (1) at (1,-1) {};
		\draw (0) to [out=-90,in=180] (1);
		\draw (1,0.5)--(1)--(1,-2);
	\end{tikzpicture}
	= \left(\begin{matrix} \psi_2 & \psi_1\\ 0 & \psi_2 \end{matrix}\right)\)
\end{center}
The only two permutations over $\mathbb C^2$ are the identity and NOT, so if we set $\psi_2 = \xi_2 = 1$, we obtain the decomposition in Eq (\ref{eqn:ldu}).
\qed
\end{proof}

Consequently, given a representative of a SLOCC-class we can reproduce the whole SLOCC-class when we augment the GHZ/W-calculus with \emph{variables}, i.e.~single-qubit states. Thus, this language is rich enough to construct any multipartite state.

\begin{corollary}\label{ConjSuper}
From the GHZ/W-pair $(\mathcal G, \mathcal W)$ and single qubit states we can obtain any $N$-qubit entangled state.
\end{corollary}

Since either Frobenius algebra generates an isomorphism with the dual space, an arbitrary $N+M$-qubit state can be used to obtain an arbitrary linear map $L:\bigotimes^N\mathbb{C}^2\to\bigotimes^M\mathbb{C}^2$.

This inductive technique is universal for constructing multipartite entangled states, so as one would expect, states built in this manner have a number of vertices that is exponential in the number of systems. However, there are often much simpler representatives of restricted classes of states, which can be expressed and manipulated in a computationally inexpensive manner. For example, the states below are in five distinct SLOCC super-classes as defined by Lamata et al in \cite{Lamatabis}.
\begin{center}
\beginpgfgraphicnamed{five_superclasses}
\InputIfFileExists{five_superclasses.tikz}{}{\input{./figures/five_superclasses.tikz}}
\endpgfgraphicnamed
\end{center}

The first two, as expected, are $\ket{GHZ_4}$ and $\ket{W_4}$. The final three are:
\begin{itemize}
\item 
$\ket{0}\underbrace{\left( \ket{000}+\ket{110}+\ket{101} \right)}_{\stackrel{\mbox{\tiny SLOCC}}{\simeq}\ketW}
+  \ket{1} (\ket{0}\underbrace{(\ket{01}+\ket{10} )}_{\stackrel{\mbox{\tiny SLOCC}}{\simeq}\ket{Bell}} )$
\item 
$\ket{0}\ket{000} +  \ket{1} (\ket{1}\underbrace{(\ket{01}+\ket{10} )}_{\stackrel{\mbox{\tiny SLOCC}}{\simeq}\ket{Bell}} )$
\item $\ket{0}\underbrace{(\ket{000}+\ket{111})}_{\stackrel{\mbox{\tiny SLOCC}}{\simeq}\ketGHZ}+\ket{1}\ket{010}$
\end{itemize}
respectively, from which we can read off the corresponding right singular vectors.	 

We can also obtain examples of fully parametrized SLOCC-superclasses. That is, the values of the variables yield all SLOCC-classes that the superclass contains. For example, the following figure corresponds with the given SLOCC superclass:
\begin{equation*}
\beginpgfgraphicnamed{class_with_params}
\InputIfFileExists{class_with_params.tikz}{}{\input{./figures/class_with_params.tikz}}
\endpgfgraphicnamed =
\ket{0}(\underbrace{(\ket{00}+\ket{1\psi})}_{
\stackrel{\mbox{\tiny SLOCC}}{\simeq}\ket{Bell}}\ket{\phi})
+
\ket{1}\ket{0}\ket{Bell}
\end{equation*}

In addition to providing an inductive method to generate arbitrary multi-partite states, the graphical calculus provides an intuitive tool for reasoning about multipartite states. Since individual components exhibit well-defined primitive behaviours via the graph rewriting, one could imagine constructing composite states to meet specific, complex behavioural specifications in a quantum algorithm or protocol.

\section{Conclusion and outlook}

In this paper, we have identified a class of highly symmetric, highly entangled tripartite states called Frobenius states. We then formulated an equivalent definition for commutative Frobenius algebras in terms of these states. Via this correspondence, we then showed that the induced tripartite state of a special commutative Frobenius algebra over $\Q$ must be SLOCC-equivalent to the GHZ state. Furthermore, any symmetric state that is SLOCC-equivalent to GHZ can be turned into a special commutative Frobenius algebra. We completed this story for tripartite qubit states by showing that the same strong relationship exists between states that are SLOCC-equivalent to the W state and anti-special commutative Frobenius algebras. Due to the exhaustiveness of the classification of tripartite states up to SLOCC in \cite{DVC}, we noted as a corollaries to our main theorem that (a) any SLOCC-maximal tripartite qubit state is a SLOCC-equivalent to a Frobenius state and (b) any commutative Frobenius algebra over qubits is locally equivalent to one that is either special or anti-special.

We take this result as strong evidence that an SCFA corresponding to the GHZ state and an ACFA corresponding to the W state should serve as the canonical building blocks of a compositional theory of multipartite states. These two algebras, subject to some conditions, enable one to design complex states and maps that provide behaviours similar to their classical circuit analogues, such as a quantum multiplexor. We prove that this and single qubit states boosts the theory GHZ/W-pairs to computational universality. We finish by showing how one builds the $n$-partite states arising from the inductive SLOCC-classification scheme of  Lamata et al in \cite{Lamata} using this language. 

Here are some concrete open questions that require further investigation:
\begin{itemize}
\item 
The obvious next step is to explore the space of states representable in this theory, and the types of (provably correct) protocols that they can implement. This may be done with the help of the {\tt quantomatic} software \cite{quantomatic}.  
\item
The conditions (i.)--(iv.) defining a GHZ/W-pair are fairly weak and by no means provide a complete characterisation of the identities present in GHZ/W-graphs with respect to the concrete GHZ and W state. Currently the only known such completeness result with respect to Hilbert spaces is Selinger's theorem for dagger compact categories \cite{SelingerCompleteness}.  Hence a substantial effort will be required to extend these conditions with other ones that will be sufficient to identify when two graphs represent the same state.
\item
One could also ask when two graphs inhabit the same equivalence class with respect to some other condition besides SLOCC (e.g. equivalence by local unitaries, or LU). In the same vein as van den Nest's theorem for graph states and the resulting  LU/LC-conjecture  \cite{LULC} (which has recently been disproved \cite{LULCfalse}), one can ask which notions of equivalence coincide for certain subsets of GHZ/W-graphs. 
\item 
The analysis in this paper is specific to qubits. The two cases that specialness and antispecialness represent correspond to $\mu\circ\delta$ being either rank 2 or rank 1. For higher dimensions, one can ask what sorts of commutative Frobenius algebras (and hence states) arise for intermediate ranks. Would it then be possible to classify the Frobenius states for qu$D$its as well?
\item
The results here may suggest new techniques for simulating many-body systems, an area which recently has seen a substantial increase of graphical methods e.g.~\cite{Pirvu}.  Note in particular the following corollary to the results in this paper:
each $D=d=2$ MPS-chain admits the following form:
\def\indexnode{\raisebox{4.8mm}{\begin{tikzpicture}[dotpic]
	\node [style=none] (0) at (-5.60, -0.40) {};
	\node [style=none] (1) at (-4.90, -0.40) {};
	\node [style=small gray dot] (2) at (-5.25, -0.75) {\tiny$i$}; 
	\node [style=none] (3) at (-5.25, -1.25) {};
	\draw  (2) to (3.north);
	\draw  (1.north) to (2);
	\draw  (0.north) to (2);\end{tikzpicture}}}
\[
\begin{tikzpicture}[dotpic]
	\node [style=none] (0) at (-6.50, 0.00) {};
	\node [style=box vertex] (1) at (-4.00, 0.00) {$V_1$};
	\node [style=box vertex] (2) at (-1.50, 0.00) {$V_2$};
	\node [style=box vertex] (3) at (1.00, 0.00) {$V_3$};
	\node [style=box vertex] (4) at (3.50, 0.00) {$V_4$};
	\node [style=small gray dot] (5) at (-5.25, -0.75) {\tiny$1$};
	\node [style=small gray dot] (6) at (-2.75, -0.75) {\tiny$2$};
	\node [style=small gray dot] (7) at (-0.25, -0.75) {\tiny$3$};
	\node [style=small gray dot] (8) at (2.25, -0.75) {\tiny$4$};
	\node [style=box vertex] (9) at (-5.25, -1.75) {$L_1$};
	\node [style=box vertex] (10) at (-2.75, -1.75) {$L_2$};
	\node [style=box vertex] (11) at (-0.25, -1.75) {$L_3$};
	\node [style=box vertex] (12) at (2.25, -1.75) {$L_4$};
	\node [style=none] (13) at (-5.25, -2.50) {};
	\node [style=none] (14) at (-2.75, -2.50) {};
	\node [style=none] (15) at (-0.25, -2.50) {};
	\node [style=none] (16) at (2.25, -2.50) {};
	\draw [out=0, looseness=1.75, in=150] (1) to (6);
	\draw [out=30, looseness=1.75, in=180] (6) to (2);
	\draw  (11) to (15.north);
	\draw [out=30, looseness=1.75, in=180] (5) to (1);
	\draw [out=0, looseness=1.75, in=150] (2) to (7);
	\draw  (9) to (13.north);
	\draw [out=30, looseness=1.75, in=180] (8) to (4);
	\draw  (8) to (12);
	\draw [out=0, looseness=1.75, in=150] (3) to (8);
	\draw  (5) to (9);
	\draw [out=0, looseness=1.75, in=150] (0.north) to (5);
	\draw  (6) to (10);
	\draw  (12) to (16.north);
	\draw [out=30, looseness=1.75, in=180] (7) to (3);
	\draw  (7) to (11);
	\draw  (10) to (14.north);
\end{tikzpicture}
\]
where $V_i, L_i$ are $2\times 2$-matrices and the nodes 
$\mu_i=\indexnode$ 
are either $\whitemult$ or $\mult$. This suggests that one could perform the bulk of the work in computing expectation values using (efficient) graphical techniques.
\end{itemize}

\end{document}